\documentclass[a4paper,UKenglish,cleveref, autoref,numberwithinsect, thm-restate]{lipics-v2021}
\newif\iflong\longtrue
\nolinenumbers
\hideLIPIcs
\usepackage{tikz}
\usetikzlibrary{matrix,backgrounds}
\usepackage[utf8]{inputenc} 
\usepackage{amsmath}
\usepackage{hyperref}
\usepackage[noend]{algpseudocode}
\usepackage{algorithm}
\usepackage{amssymb}
\usepackage{todonotes}
\usepackage{dsfont}

\usepackage[]{hyperref}
\usepackage{cleveref}

\usepackage{booktabs}

\newtheorem{rrule}[theorem]{Reduction Rule}

\crefname{rrule}{Reduction Rule}{Reduction Rules}

\newcommand{\Oh}{\ensuremath{\mathcal{O}}}
\DeclareMathOperator{\cl}{cl}
\DeclareMathOperator{\cp}{cap}

\DeclareMathOperator{\argmax}{argmax}
\DeclareMathOperator{\opt}{opt}

\newcommand{\p}{P} 
\newcommand{\q}{Q} 

\newcommand{\DS}{{\normalfont\textsc{Dominating Set}}}
\newcommand{\IS}{{\normalfont\textsc{Independent Set}}}
\newcommand{\IM}{{\normalfont\textsc{Induced Matching}}}
\newcommand{\ConVC}{{\normalfont\textsc{Connected Vertex Cover}}}
\newcommand{\CapVC}{{\normalfont\textsc{Capacitated Vertex Cover}}}
\newcommand{\Ramsey}{{\normalfont\textsc{Ramsey}}}

\newcommand{\PHC}{{\normalfont coNP $\subseteq$ NP/poly}} 
\newcommand{\nPHC}{{\normalfont coNP $\not \subseteq$ NP/poly}} 

\newcommand{\problemdef}[3]{
  \begin{center}
    \begin{minipage}{0.95\textwidth}
      \normalsize\textsc{#1} \smallskip \\
      \begin{tabularx}{\textwidth}{@{}l@{\hspace{3pt}}X}
        \normalsize\textbf{Input:}    & \normalsize#2 \\
        \normalsize\textbf{Question:} & \normalsize#3
      \end{tabularx}
    \end{minipage}
  \end{center}
}

\tikzstyle{vertex}=[draw, circle, fill, inner sep = 2.4pt]

\title{Essentially Tight Kernels for\\ (Weakly) Closed Graphs}
\titlerunning{Essentially Tight Kernels for (Weakly) Closed Graphs}

\author{Tomohiro Koana}{Technische Universität Berlin, Algorithmics and Computational Complexity, Germany}{tomohiro.koana@tu-berlin.de}{https://orcid.org/0000-0002-8684-0611}{Supported by the Deutsche Forschungsgemeinschaft (DFG), project FPTinP, NI 369/19.}

\author{Christian Komusiewicz}{Philipps-Universität Marburg, Fachbereich Mathematik und Informatik,  Marburg, Germany}{komusiewicz@informatik.uni-marburg.de}{https://orcid.org/0000-0003-0829-7032}{}

\author{Frank Sommer}{Philipps-Universität Marburg, Fachbereich Mathematik und Informatik,  Marburg, Germany}{fsommer@informatik.uni-marburg.de}{https://orcid.org/0000-0003-4034-525X}{Partially supported by the Deutsche Forschungsgemeinschaft (DFG), project MAGZ, KO~3669/4-1.}

\authorrunning{Tomohiro~Koana, Christian~Komusiewicz, Frank~Sommer}




\begin{document}

\maketitle

\begin{abstract}
  We study kernelization of classic hard graph problems when the input graphs fulfill triadic closure properties. More precisely, we consider the recently introduced parameters closure
  number~$c$ and the weak closure number~$\gamma$ [Fox et al., SICOMP 2020] in
  addition to the standard parameter solution size~$k$. For \textsc{Capacitated Vertex
    Cover}, \textsc{Connected Vertex Cover}, and \textsc{Induced Matching} we obtain the
  first kernels of size~$k^{\Oh(\gamma)}$ and~$(\gamma k)^{\Oh(\gamma)}$, respectively, thus
  extending previous kernelization results on degenerate
  graphs. 
  The kernels are essentially tight, since these problems are unlikely to admit kernels of size $k^{o(\gamma)}$ by previous results on their kernelization complexity in degenerate graphs [Cygan et al., ACM TALG 2017].
  In addition, we provide lower bounds for the kernelization of
  \textsc{Independent Set} on graphs with constant closure number~$c$ and kernels for \textsc{Dominating Set} on weakly closed split graphs and weakly closed bipartite graphs.
\end{abstract}

\section{Introduction}


A main tool for coping with hard computational problems is to design polynomial-time data
reduction algorithms which shrink large input data to a computationally hard core by
removing the easy parts of the instance. Parameterized algorithmics provides
the framework of \emph{kernelization} for analyzing the power and limits of
polynomial-time data reduction algorithms.

A parameterized problem comes equipped with the classic input instance~$I$ and a
parameter~$k$ which describes the structure of the input data or is a bound on the
solution size. A kernelization is an algorithm that replaces every input instance~$(I,k)$
in polynomial time by an equivalent instance~$(I',k')$ (the kernel) whose size depends
only on the parameter~$k$, that is, where~$|I'| + k'\le g(k)$ for some computable function~$g$. The kernel is
guaranteed to be small if~$k$ is small and~$g$ grows only modestly. A
particularly important special case are thus \emph{polynomial kernelizations} where~$g$ is
a polynomial.

Many problems do not admit a kernel simply because they are believed to be not
\emph{fixed-parameter tractable}. That is, it is assumed that they are not solvable
in~$f(k)\cdot |I|^{\Oh(1)}$~time. A classic example for such a problem is
\textsc{Independent Set} parameterized by the solution size~$k$. Moreover, even many
problems that do admit kernels are known to not admit polynomial kernels~\cite{BDFH09,DLS14,Kra14,Kra14b}\footnote{All kernelization lower bounds mentioned in this work are based on the assumption \nPHC.}; a classic example
is the \textsc{Connected Vertex Cover} problem parameterized by the solution
size~$k$~\cite{DLS14}.
To devise and analyze kernelization algorithms for such problems, one considers either
further additional parameters or restricted classes of input graphs. One example for this
approach is the study of kernelization in degenerate graphs~\cite{CGH17,CPPW12,PRS12}. A
graph~$G$ is~$d$-degenerate if every subgraph of~$G$ contains at least one vertex that has
degree at most~$d$. The \textsc{Dominating Set} problem, for example, is assumed to be
fixed-parameter intractable with respect to the solution size~$k$ in general graphs, but
admits a kernel of size~$k^{\Oh(d^2)}$ where~$d$ is the degeneracy of the input
graph~\cite{PRS12}. Thus, the degree of the kernel size depends on~$d$; we will say that
\textsc{Dominating Set} admits a polynomial kernel on $d$-degenerate graphs. In subsequent
work, this kernelization was shown to be tight in the sense that there is no kernel of
size~$k^{o(d^2)}$~\cite{CGH17}. The situation is different for \textsc{Independent Set} which admits a trivial problem kernel with~$\Oh(dk)$ vertices: here the kernel size is polynomial
in~$d+k$.

Real-world networks have small degeneracy, making it an interesting parameter from an
application point of view. Moreover, bounded degeneracy imposes an interesting combinatorial
structure that can be exploited algorithmically as evidenced by the discussion above.
Recently, Fox et al.~\cite{FRSWW20} discovered two further parameters that share these two
features: they are well-motivated from a practical standpoint and describe interesting and useful
combinatorial features of graphs.
The first parameter is the \emph{closure} of a graph, defined as follows.
\begin{definition}[\cite{FRSWW20}]
  \label{def:c}Let $\cl_G(v) = \max_{w \in V \setminus N[v]} |N(v) \cap N(w)|$ denote
  the \emph{closure number} of a vertex~$v$ in a graph~$G$. A graph $G$ is
  \emph{$c$-closed} if $\cl_G(v) < c$ for all $v \in V(G)$.
\end{definition}

In other words, a graph is~$c$-closed if every pair of nonadjacent vertices has at
most~$c-1$ common neighbors. The idea behind the parameter is to model triadic closure in social networks: the observation that it is unlikely that two persons have
many common acquaintances without knowing each other. Fox et al.~\cite{FRSWW20} devised a
further parameter, the \emph{weak closure number} which relates to $c$-closure
in the same way that degeneracy relates to maximum degree: instead of demanding a bounded
closure number for every vertex, we only demand that every subgraph
 contains at least one vertex with bounded closure number.
\begin{definition}[\cite{FRSWW20}]
  \label{def:gamma}
  A graph $G$ is weakly $\gamma$-closed\footnote{To avoid confusion with the closure number~$c$, we denote the weak closure by~$\gamma$ instead of~$c$.} if
  \begin{itemize}
    \item
      there exists a \emph{closure ordering} $\sigma:= v_1, \dots, v_n$ of the vertices such that $\cl_{G_i}(v_i) < \gamma$ for all $i \in [n]$ where $G_i := G[\{ v_i, \dots, v_n \}]$, or, equivalently, if 
    \item
      every induced subgraph $G'$ of $G$ has a vertex $v \in V(G')$ such that $\cl_{G'}(v) < \gamma$.
  \end{itemize}

\end{definition}

We call~$\sigma$ a \emph{closure ordering} of~$G$. The three parameters are related as
follows: The weak closure number~$\gamma$ is never larger than~$d+1$ and also never larger
than the closure number~$c$ of~$G$. Moreover,~$\gamma$ can be much smaller than~$d$ as
witnessed by large cliques. The degeneracy~$d$ and the closure~$c$ are unrelated as
witnessed by large cliques (these are 1-closed and have large degeneracy) and complete
bipartite graphs where one part has size two (these are 2-degenerate and have large
closure number). The latter example also shows that~$\gamma$ can be much smaller
than~$c$ which is very often the case in real-world data~\cite{FRSWW20,KKS20B}.

Akin to degeneracy, $c$-closure and weak~$\gamma$-closure have been proven to be very
useful parameters. In particular, all maximal cliques of a graph can be enumerated
in~$3^{\gamma/3}\cdot n^{\Oh(1)}$ time~\cite{FRSWW20}. By the discussion on the relation of~$\gamma$
and~$d$ above, this result thus extends the range of tractable clique enumeration
instances from the class of bounded-degeneracy graphs to the larger class of 
graphs with bounded weak closure. The clique enumeration algorithm has also
been extended to the enumeration of other cliquish subgraphs~\cite{HR20,KKS20a}

In previous work, we studied the kernelization complexity of classic problems such as
\textsc{Independent Set}, \textsc{Dominating Set}, and \textsc{Induced Matching}
on~$c$-closed graphs. This work continues this study and, more importantly, is the first
to focus on kernelization in weakly~$\gamma$-closed graphs. The main questions that we
want to address are the following: Can we improve polynomial kernelizations on
$d$-degenerate graphs or on $c$-closed graphs to polynomial kernelizations on
weakly~$\gamma$-closed graphs? What are the structural features of weakly~$\gamma$-closed
graphs that we can exploit in kernelization algorithms? What are the limits of
kernelization on (weakly) closed graphs?

\subparagraph{Our Results.} Building on a combinatorial lemma of Frankl and
Wilson~\cite{FW81}, we obtain a general lemma (\Cref{thm:bound-sets-p-and-q}) which can be
used to bound the size of graphs in terms of their vertex cover number and weak closure
number. More precisely, we show that in a graph~$G$ with vertex cover~$S$ of size~$k$ and
weak closure~$\gamma$, the number of different neighborhoods in the independent
set~$I:=V(G)\setminus S$ is~$k^{\Oh(\gamma)}$.  \Cref{thm:bound-sets-p-and-q} gives a general strategy for
obtaining kernels in weakly closed graphs: Devise reduction rules that 1) bound the size
of the vertex cover and 2) decrease the size of neighborhood classes.

We then show that this strategy helps in obtaining kernels on weakly closed graphs for
\textsc{Capacitated Vertex Cover}, \textsc{Connected Vertex Cover}, and \textsc{Induced
  Matching} all parameterized by the natural parameter solution size~$k$. For these problems,
polynomial kernels in degenerate graphs were known~\cite{CGH17,CPPW12,EKKW10,KPSX11}. Our results thus extend the class of graphs for which
polynomial kernels are known for these problems. The kernels have size~$k^{\Oh(\gamma)}$
and~$(\gamma k)^{\Oh(\gamma)}$, respectively, and it follows from previous results on
degenerate graphs that the dependence on~$\gamma$ in the exponent cannot be
avoided~\cite{CGH17,CPPW12}.

We complement these findings with a study of \textsc{Capacitated Vertex Cover} and \textsc{Connected Vertex Cover} on~$c$-closed graphs. Interestingly, the kernelization complexity of the  problems differs: \textsc{Capacitated Vertex Cover} does not admit a kernel of size~$\Oh(k^{(c-1)/2-\epsilon})$ whereas \textsc{Connected Vertex Cover} admits a kernel with $\Oh(ck^2)$~vertices.

Next, we revisit the kernelization complexity of \textsc{Independent Set} on (weakly)
$c$-closed graphs. We  show that \textsc{Independent Set} does not admit a
kernel of size~$\Oh(k^{2-\epsilon})$ even on $c$-closed graphs where~$c$ is constant. This
complements previous kernels of size~$\Oh(c^2k^3)$~\cite{KKS20}
and~$\Oh(\gamma^2k^3)$~\cite{KKS20a} and narrows the gap between upper and lower bound for
the achievable kernel size on (weakly) closed graphs.
We also obtain a lower bound of~$\Oh(k^{4/3-\epsilon})$ on the number of \emph{vertices}
in case of constant~$c$ and show that at least a linear dependence on~$c$ is necessary in
any kernelization of \textsc{Independent Set} in (weakly) closed graphs: under standard
assumptions, \textsc{Independent Set} does not admit a kernel of size
$c^{(1-\epsilon)}\cdot k^{\Oh(1)}$. Some of our results also hold for Ramsey-type problems
where one wants to find a large subgraph belonging to a class~$\mathcal{G}$ containing all
complete and all edgeless graphs. In this context, we observe that weakly
$\gamma$-closed graphs fulfill the Erdős–Hajnal property~\cite{EH89} with a linear dependence
on~$\gamma$: There is a constant~$q$ such that every weakly $\gamma$-closed graph
on~$k^{q\gamma}$ vertices has either a clique of size~$k$ or an independent set of
size~$k$. We believe that this observation is of independent interest and that it will be useful in
the further study of weakly $\gamma$-closed graphs.

Finally, we consider \textsc{Dominating Set} for which we developed a kernel of
size~$k^{\Oh(c)}$ in previous work~\cite{KKS20}. The natural open question is whether
\textsc{Dominating Set} admits a kernel of size~$k^{f(\gamma)}$ for some function~$f$, which would extend the class of
kernelizable input graphs from degenerate to weakly closed. We make partial progress
towards answering this question by showing that \textsc{Dominating Set} admits a kernel of
size~$k^{\Oh(\gamma^2)}$ in weakly closed graphs with constant clique number (these
include bipartite graphs) and a kernel of size~$(\gamma k)^{\Oh(\gamma)}$ in split graphs. In both cases the dependence on~$\gamma$ is tight in the sense that kernels of size~$k^{o(d^2)}$ and of size~$k^{o(c)}$ are unlikely to exist~\cite{CGH17,KKS20a}.


\subparagraph{Preliminaries}

By~$[n]$ we denote the set $\{ 1, \dots, n \}$ for some~$n \in \mathds{N}$.
For a graph~$G$, we denote by~$V(G)$ its \emph{vertex set} and by~$E(G)$ its \emph{edge set}.
By~$n := |V(G)|$ we denote the number of vertices.
Let~$X \subseteq V(G)$ be a vertex set.
By~$G[X]$ we denote the \emph{subgraph induced} by~$X$ and by~$G - X := G[V(G) \setminus X]$ we denote the graph obtained by removing the vertices of~$X$.
We denote by~$N_G(X):=\{ y \in V(G) \setminus X \mid xy \in E(G), x \in X \}$ the \emph{open neighborhood} of~$X$ and by~$N_G[X]:=N_G(X) \cup X$ the \emph{closed neighborhood} of~$X$.
The maximum degree of~$G$ is~$\Delta_G := \max_{v \in V(G)} \deg_G(v)$. 
A vertex with degree one is referred to as a \emph{leaf} vertex.
Let~$\sigma$ be a closure ordering.
We define~$\p_G^\sigma(v) := \{ u \in N_G(v) \mid \text{$u$ appears before $v$ in $\sigma$} \}$ and~$\q_G^\sigma(v) := \{ u \in N_G(v) \mid \text{$u$ appears after $v$ in $\sigma$} \}$.
We say that~$P(v)$ are \emph{prior neighbors} of~$v$ and~$Q(v)$ are \emph{posterior neighbors} of~$v$.
A \emph{matching}~$M$ is a set of edges without common vertices.
By~$V(M)$ we denote the union of all endpoints of edges in~$M$.
We omit the superscripts and subscripts when they are clear from the context.
Note that~$|Q(u) \cap N(v)| < \gamma$ for any nonadjacent vertices~$u, v \in V(G)$.
This leads to the following observation.

\begin{observation}
\label{obs-intersection-q-neighbors}
  For nonadjacent vertices $u, v \in V(G)$, it holds that $|Q(u) \cap Q(v)| \le |Q(u) \cap N(v)| \le \gamma - 1$.
\end{observation}

\iflong A parameterized problem is \emph{fixed-parameter tractable} if every instance~$(I, k)$ can be solved in~$f(k) \cdot |I|^{\Oh(1)}$ time for some computable function~$f$.
An algorithm with such a running time is an \emph{FPT algorithm}.
A \emph{kernelization} is a polynomial-time algorithm which transforms every instance~$(I, k)$ into an equivalent instance~$(I', k')$ such that~$|I'| + k' \le g(k)$ for some computable function~$g$.
It is widely believed that W[$t$]-hard problems~$(t \in \mathds{N})$ do not admit an FPT algorithm. \fi
For more details on parameterized complexity, we refer to the standard monographs~\cite{CFK+15,DF13}.

\section{Variants of Vertex Cover}
Frankl and Wilson~\cite{FW81} proved the following bound on the size of set systems where the number of different intersection sizes is bounded. 
\begin{proposition}[Frankl and Wilson \cite{FW81}]
  \label{prop:fw81}
  Let $\mathcal{F}$ be a collection of pairwise distinct subsets of $[n]$ and let $L \subseteq [n]$ be some subset.
  If $|S \cap S'| \in L$ for all distinct $S, S' \in \mathcal{F}$, then~$|\mathcal{F}| \in \Oh(n^{(|L|)})$.
\end{proposition}
We now use this lemma to achieve a bound on the number of neighborhoods in weakly closed graphs with a small vertex cover.
\begin{lemma}
\label{thm:bound-sets-p-and-q}
  Let $G$ be a graph and let $I$ be an independent set of $G$.
  Suppose that for every vertex $v \in I$, there are at most $t$ vertices $v' \in I \setminus \{ v \}$ such that $N(v) = N(v')$.
  Then,~$|I| \in t \cdot \Oh(3^{\gamma/3}\cdot k^{2\gamma+3})$, where $k := n - |I|$.
\end{lemma}
\begin{proof}
  Let $S := V(G) \setminus I$.
  We say that two vertices $v, v' \in I$ are \emph{$P$-equivalent}, \emph{$Q$-equivalent}, and \emph{$N$-equivalent} if $P(v) = P(v')$, $Q(v) = Q(v')$, and $N(v) = N(v')$, respectively.
  Let~$\mathcal{P}$,`$\mathcal{Q}$, and $\mathcal{N}$ denote the collection of $P$-equivalence, $Q$-equivalence, and $N$-equivalence classes, respectively. We extend the notation of~$P$,~$Q$, and~$N$ to equivalence classes~$A$ by defining~$P(A):=P(v)$,~$Q(A):=Q(v)$, and~$N(A):=N(v)$ for some~$v\in A$.
  Since there is at most one $N$-equivalence class for every $P$-equivalence class and $Q$-equivalence class, we have~$|\mathcal{N}| \le |\mathcal{P}| \cdot |\mathcal{Q}|$.
  By the assumption that there are at most $t$ vertices in each $N$-equivalence class, we also have $|I| \le t \cdot |\mathcal{N}|$.
  Thus, it suffices to show suitable bounds on $|\mathcal{P}|$ and~$|\mathcal{Q}|$.

  First, we prove that $|\mathcal{Q}| \in \Oh(k^{\gamma})$, using the result of Frankl and Wilson (\Cref{prop:fw81}~\cite{FW81}).
  To do so, we associate each $Q$-equivalence class $A$ with the set $Q(A)$. Since~$I\supseteq A$ is an independent set, $Q(A) \subseteq S$.
  Moreover, for two distinct $Q$-equivalence classes $A$ and $A'$, we have $|Q(A) \cap Q(A')| < \gamma$ by Observation~\ref{obs-intersection-q-neighbors}, and equivalently, $|Q(A) \cap Q(A')| \in L$ for $L := [\gamma - 1]$.
  Consequently, by \Cref{prop:fw81}, we have $|\mathcal{Q}| \in \Oh(|S|^{|L|}) = \Oh(k^{\gamma})$.

  Next, we bound the size of $\mathcal{P}$.
  Let $I_0 := \{ v \in I \mid \exists u, w \in P(v) \colon uw \notin E(G) \}$ be the set of vertices in $I$ with nonadjacent prior neighbors.
  By the definition of weak $\gamma$-closure, there are at most $\gamma- 1$ vertices of $I_0$ for every pair of nonadjacent vertices in $S$.
  Thus, we have~$|I_0| < \gamma \binom{|S|}{2} \in \Oh(\gamma k^2)$.

  Let $I_1 := I \setminus I_0$ and let $\mathcal{P}_1$ be the collection of $P$-equivalence classes in $I_1$.
  Note that for every $A \in \mathcal{P}_1$, its neighborhood $P(A)$ is a clique.
  Since a weakly $\gamma$-closed graph on~$n$ vertices has $\Oh(3^{\gamma / 3} n^2)$ maximal cliques \cite{FRSWW20}, there are $\Oh(3^{\gamma / 3} k^2)$ equivalence classes~$A$ such that~$P(A)$ constitutes a maximal clique in $G[S]$.
  Consider an equivalence class~$A$ such that~$P(A) \subset C$ for some maximal clique $C$ in~$G[S]$.
  We will show that there are $k^{\Oh(\gamma)}$~such equivalence classes.
  Let $u$ be the first vertex of $C \setminus P(A)$ in the closure ordering~$\sigma$.
  Since $P(A) \subset C \subseteq N(u) = P(u) \cup Q(u)$, we have $P(A) = (P(A) \cap P(u)) \cup (P(A) \cap Q(u))$.
  As~$P(A) \cap P(u) = C \cap P(u)$ by the choice of $u$, we can rewrite $P(A) = (C \cap P(u)) \cup B$, where~$B := P(A) \cap Q(u)$.
  Thus, there is at most one equivalence class of $\mathcal{P}_1$ for every maximal clique $C$ in $G[S]$, vertex~$u \in S$, and vertex subset $B \subseteq S$, and thereby, we have $|\mathcal{P}_1| \in \Oh(3^{\gamma / 3} k^2 \cdot k \cdot b)$, where $b$ denotes the number of choices for $B$.
  Observe that~$P(A) = P(v)$ for some vertex $v \in I_1$ and thus that $B = Q(u) \cap P(v) \subseteq Q(u) \cap N(v)$.
  It follows that $|B| \le |Q(u) \cap N(v)| < \gamma - 1$ by Observation~\ref{obs-intersection-q-neighbors}, and hence $b \in \Oh(k^{\gamma})$ and $|\mathcal{P}_1| \in \Oh(3^{\gamma/3}\cdot k^3 \cdot k^{\gamma})= \Oh(3^{\gamma/3}\cdot k^{\gamma+3})$.
  Overall, we have $|\mathcal{P}| \le (|I_0| + |\mathcal{P}_1|) \in \Oh(3^{\gamma/3}\cdot k^{\gamma+3})$. The total number of~$N$-equivalence classes is thus at most $|\mathcal{Q}|\cdot |\mathcal{P}|\in \Oh(3^{\gamma/3}\cdot k^{2\gamma+3})$.
\end{proof}
\subsection{Capacitated Vertex Cover}
The first problem to which we apply \Cref{thm:bound-sets-p-and-q} is \textsc{Capacitated Vertex Cover}.

\problemdef {Capacitated Vertex Cover} {A graph $G$, a capacity function
  $\cp \colon V(G) \to\mathds{N}$, and $k \in \mathds{N}$.} {Is there a set $S$ of at most
  $k$ vertices and a function $f$ mapping each edge of~$E(G)$ to one of its endpoints in~$S$ such that
  $|\{ e \in E(G) \mid f(e) = v \}| \le \cp(v)$ for all $v \in S$?}  
  
  Cygan et
al.~\cite{CGH17} showed that \textsc{Capacitated Vertex Cover} admits a kernel with
$\Oh(k^{d + 1})$~vertices. Furthermore, they proved that this kernel is essentially tight:
a kernel with $\Oh(k^{d-\epsilon})$~vertices would imply \PHC~\cite{CGH17}.
We will show that the same reduction rule also gives a kernel in graphs with bounded weak closure. One may view this result as a way of showing that the rule is more powerful than what was previously known.

The kernel uses the following rule.
\begin{rrule}[\cite{CGH17}]
  \label{rr:twincrown}
  If $S \subseteq V(G)$ is a subset of twin vertices with a common neighborhood $N(S)$ such that $|S| = k + 2 \ge |N(S)|$, then remove a vertex with minimum capacity in $S$ from $G$, and decrease all the capacities of vertices in $N(S)$ by one.
\end{rrule}

We omit the proof for the correctness of \Cref{rr:twincrown}, referring to Cygan et al.~\cite[Lemma 20]{CGH17}.
One can easily verify that \Cref{rr:twincrown} does not increase the weak~$\gamma$-closure.
In the following theorem, we show that \Cref{rr:twincrown} indeed gives us a kernel with $k^{\Oh(\gamma)}$~vertices.

\begin{theorem}
  \label{thm:capvc}
  \CapVC{} has a kernel of size $k^{\Oh(\gamma)}$.
\end{theorem}
\begin{proof}
  We show that a yes-instance which is reduced with respect to \Cref{rr:twincrown} has
  size~$k^{\Oh(\gamma)}$. Let~$S$ be a capacitated vertex cover of size at most~$k$
  of~$(G,\cp,k)$. Let~$I:=V(G)\setminus S$. By definition,~$I$ is an independent set
  and~$N(v)\subseteq S$ for all~$v\in I$. Moreover, since~$(G,\cp,k)$ is reduced with
  respect to~\Cref{rr:twincrown} there is no set of~$k+2$ vertices in~$I$ that have the
  same neighborhood. Hence,~$I$ fulfills the condition of \Cref{thm:bound-sets-p-and-q} with~$t=k+2$. Thus,~$|I|\in k\cdot k^{\Oh(\gamma)}$ which implies~$|V(G)|=|S|+|I|\in k^{\Oh(\gamma)}$.
\end{proof}

We also show that this kernel is
essentially tight even if~$\gamma$ is replaced by~$c$.

\begin{theorem}
\label{thm:capvc-lb-for-c}
For~$c \ge 4$, \CapVC{} has no compression of size~$\Oh(k^{\frac{c-1}{2}-\epsilon})$ unless \PHC.
\end{theorem}

\iflong

\begin{proof}
We will show the theorem by a reduction from \textsc{$\lambda$-Exact Set Cover}. 

\problemdef
{$\lambda$-Exact Set Cover}
{A set family~$\mathcal{F}$ over an universe $U:=\{u_1, u_2, \ldots , u_{\lambda k}\}$ such that each~$F\in\mathcal{F}$ has size~$\lambda$, and $k\in\mathds{N}$.}
{Is there a subfamily~$\mathcal{S} \subseteq \mathcal{F}$ of size exactly~$k$ such that~$\bigcup_{S\in\mathcal{S}}=U$?}

For any~$\lambda \ge 3$, \textsc{$\lambda$-Exact Set Cover} does not have a compression of size~$\Oh(k^{\lambda-\epsilon})$ unless \PHC~\cite{DM12,HW12}. 
Let~$(U,\mathcal{F},k)$ be an instance of \textsc{$\lambda$-Exact Set Cover}. 
We will construct a~$(2\lambda+1)$-closed graph~$G$ as follows: 
The vertex set~$V(G)$ consists of one copy of~$\mathcal{F}$, and two copies~$U^1$ and~$U^2$ of~$U$. 
Furthermore, one leaf-vertex is attached to each vertex in~$U^1\cup U^2$.  
The edges between~$\mathcal{F}$ and both copies~$U^1$ and~$U^2$ of~$U$ represent the incidence graph of the instance~$(U,\mathcal{F},k)$.
Furthermore, we add edges such that~$U^1\cup U^2$ forms a clique in~$G$.
Note that~$G$ contains exactly~$2\lambda k +2\lambda|\mathcal{F}|+ \lambda k (2\lambda k-1)$ edges. 
Next, we set the capacity of each leaf-vertex which is attached to a vertex in~$U^1\cup U^2$ to zero, and the capacity of each vertex in~$\mathcal{F}$ to~$2\lambda$. 
For an element~$u_i\in U$, we denote by~$z_i$ the number of sets in~$\mathcal{F}$ containing element~$u_i$.
We set the capacity of the~$i$-th vertex~$u^1_i$ of~$U^1$ to~$z_i+2\lambda k-i$ and the capacity of the~$i$-th vertex~$u^2_i$ of~$U^2$ to~$z_i+i-1$. 
Finally, we set~$k':=2\lambda k+k$.
Since~$\deg(F)=2\lambda$ for each~$F\in\mathcal{F}$, each leaf-vertex has degree one, and~$U^1\cup U^2$ is a clique, we observe that~$G$ is indeed~$2\lambda+1$-closed. 

To avoid heavy notation we will not give an exact definition of the mapping function~$f$. 
Instead, we say that an edge~$uv$ is \emph{mapped} to one of its endpoints~$u$.
Formally this means that~$f(uv)=u$.

Let~$\mathcal{S}$ be an exact set cover. We prove that~$X:=\mathcal{S}\cup U^1\cup U^2$ is a capacitated vertex cover of~$G$. 
Since~$N[U^1\cup U^2]=V(G)$, the set~$X$ is a vertex cover of size exactly~$2\lambda k+k$ of~$G$. 
It remains to show that there is a mapping from each edge to one of its endpoints in~$X$. 
Edges between vertices in~$\mathcal{F}$ and~$U^1$ are mapped to the vertex in~$\mathcal{F}$ if and only if this vertex is contained in~$\mathcal{S}$. 
Otherwise, this edge is mapped to the corresponding vertex in~$U^1$. 
We map edges between~$\mathcal{F}$ and~$U^2$ analogously. 
Edges between~$U^1$ and~$U^2$ are mapped to the vertex in~$U^1$ and edges where one endpoint is a leaf-vertex are mapped to the other endpoint of that edge.
Furthermore, an edge between two vertices in~$U^1$ is mapped to the vertex with lower index and similar, an edge between two vertices in~$U^2$ is mapped to the vertex with higher index. 
Hence, each vertex in~$\mathcal{S}$ covers exactly~$2\lambda$ edges and since~$\mathcal{S}$ is an exact cover, exactly one edge incident with a vertex in~$U^1\cup U^2$ is covered by~$\mathcal{S}$. 
Thus, the~$i$-th vertex~$u^1_i\in U^1$ covers exactly~$\lambda k-i$ edges to other vertices in~$U^1$, exactly one edge to a leaf-vertex, exactly~$\lambda k$ edges to vertices in~$U^2$, and exactly~$z_i-1$ edges to vertices in~$\mathcal{F}$. 
These are exactly~$z_i+2\lambda k-i$ many. 
Similarly, the~$i$-th vertex~$u^2_i\in U^2$ covers exactly~$i-1$ edges to other vertices in~$U^2$, exactly one edge to a leaf-vertex, and exactly~$z_i-1$ edges to vertices in~$\mathcal{F}$. 
These are exactly~$z_i+i-1$ many. 
Hence,~$\mathcal{S}\cup U^1\cup U^2$ is indeed a capacitated vertex cover of~$G$. 

Conversely, suppose that~$G$ has a capacitated vertex cover~$X$ of size~$2\lambda k+k$.
Since the capacity of each leaf-vertex is zero, we observe that~$U^1\subseteq X$, and~$U^2\subseteq X$.
Hence, for the set~$\mathcal{S}:=\mathcal{F}\cap X$ we can assume that~$|\mathcal{S}|=k$. 

Recall that for an element~$u_i$ the number~$z_i$ denotes the number of sets in~$\mathcal{F}$ containing~$u_i$. 
Hence,~$\sum_{i=1}^{\lambda k}z_i=\lambda|\mathcal{F}|$. 
Furthermore, since each vertex in~$\mathcal{F}$ has capacity exactly~$2\lambda$, we observe that the total capacity of all vertices in~$X$ is~$2\lambda|\mathcal{F}|+2\lambda k+ \lambda k (2\lambda k-1)$. 
Since this matches the number of edges in~$G$, we conclude the following.

\begin{observation}
\label{obs-total-capcity-of-x}
The number of edges mapped to a vertex in~$X$ is equal to its capacity.
In particular, each edge with one endpoint in~$S\in\mathcal{S}$ is mapped to~$S$.
\end{observation}

In the following, we prove inductively over the vertex indices~$i\in[\lambda k]$ \textbf{a)} that element~$u_i$ is covered exactly once, and \textbf{b)} that the edge~$u^1_iu^1_j$ in~$U^1$ is mapped to the vertex with lower index and that the edge~$u^2_iu^2_j$ in~$U^2$ is mapped to the vertex with higher index.
Afterwards, by setting~$i:=\lambda k$ we can conclude that~$\mathcal{S}$ is an exact cover for~$U$.
Assume towards a contradiction of \textbf{a)} that element~$u_i$ is not covered exactly once by~$\mathcal{S}$. 

\textbf{Case~1:} Element~$u_i$ is not covered by~$\mathcal{S}$. 
We give a lower bound on the number of edges in~$G$ which have to be mapped to vertex~$u^2_i$.
Since~$u_i$ is not covered by~$\mathcal{S}$, all~$z_i$ edges with the endpoint~$u^2_i$ and the other endpoint in~$\mathcal{F}$ have to be mapped to~$u^2_i$. 
Furthermore, since each leaf-vertex has capacity zero, also the edge from the leaf-vertex attached to~$u^2_i$ has to be mapped to~$u^2_i$. 
First, assume that~$i=1$. 
Then, at least~$z_1+1$ edges are mapped to vertex~$u^2_1$, contradicting the fact that vertex~$u^2_1$ has capacity~$z_1$.
Second, assume that~$i\ge 2$. 
According to the induction hypothesis of \textbf{b)} for each~$j<i$ the edge~$u^2_iu^2_j$ in~$U_2$ is mapped to~$u^2_i$.
Hence, overall at least~$z_i+i$ edges are mapped to~$u^2_i$.
This is not possible since vertex~$u^2_i$ has capacity~$z_i+i-1$.

\textbf{Case~2:} Element~$u_i$ is covered at least twice by~$\mathcal{S}$. 
We upper bound the number of edges in~$G$ which can be mapped to vertex~$u^1_i$. 
Since each leaf-vertex has capacity zero, the edge from the leaf-vertex attached to~$u^1_i$ is mapped to vertex~$u^1_i$.
Furthermore, all edges between~$u^1_i$ and vertices of~$U^2$ can be mapped to~$u^1_i$.
Also, since~$u_i$ is covered at least twice by~$\mathcal{S}$, according to Observation~\ref{obs-total-capcity-of-x}, we observe that at most~$z_i-2$ edges between~$\mathcal{F}$ and vertex~$u^1_i$ can be mapped to~$u^1_i$. 
First, assume that~$i=1$. 
Observe that all~$\lambda k-1$ many edges with one endpoint~$u^1_1$ and the other endpoint in~$U^1$ can be mapped to~$u^1_1$. 
Hence, at most~$1+\lambda k+z_1-2+\lambda k-1=z_1+2\lambda k -2$ edges can be mapped to~$u^1_1$. 
According to Observation~\ref{obs-total-capcity-of-x} this is not possible, since vertex~$u^1_1$ has capacity~$z_1+2\lambda k-1$. 
Second, assume that~$i\ge 2$. 
According to the induction hypothesis of \textbf{b)} for each~$j<i$ the edge~$u^1_iu^1_j$ in~$U_1$ is mapped to~$u^1_j$.
Hence, at most~$\lambda k-i$ edges within~$U^1$ are mapped to vertex~$u^1_i$. 
Thus, overall at most~$1+\lambda k+z_i-2+\lambda k-i=z_i+2\lambda k-i-1$ many edges are mapped to~$u^1_i$.
Since the capacity of~$u^1_i$ is exactly~$z_i+2\lambda k -i$ this contradicts Observation~\ref{obs-total-capcity-of-x}. 

Hence, element~$u_i$ is covered exactly once. 
Next, we show that this implies \textbf{b)}. 
To match the capacity~$z_i+2\lambda k -i$ of vertex~$u^1_i$ the edge from the leaf-vertex attached to~$u^1_i$, all edges between~$u^1_i$ and a vertex of~$U_2$, and because of \textbf{a)} exactly~$z_i-1$ edges between vertex~$u^1_i$ and a vertex of~$\mathcal{F}$ have to be mapped to vertex~$u^1_i$. 
These are exactly~$z_i+\lambda k$ many. 
By induction hypothesis of \textbf{b)} we know that all edges of the form~$u^1_iu^1_j$ with~$j<i$ are mapped to vertex~$u^1_j$. 
Hence, exactly~$\lambda k -i$ edges within~$U^1$ are not mapped yet. 
These edges have the form~$u^1_iu^1_j$ for~$j\in[i+1,\lambda k]$. 
Because of Observation~\ref{obs-total-capcity-of-x} all these edges have to be mapped to vertex~$u^1_i$.
By similar arguments it follows that all edges of the form~$u^2_iu^2_j$ for~$j\in[i+1,\lambda k]$ are mapped to~$u^2_j$. 
Hence, \textbf{b)} is proved. 
As mentioned above, by setting~$i=\lambda k$ we conclude from \textbf{a)} that~$\mathcal{S}$ is an exact cover for~$(U,\mathcal{F},k)$.

Observe that since~$\lambda$ is a constant we obtain~$k'\in\Oh(k)$. 
Thus, it follows from the result of Hermelin and Wu~\cite{HW12} that if \CapVC{} admits a kernel of size~$\Oh(k^{(c-1)/2-\epsilon})$ for some~$\epsilon >0$ in~$c$-closed graphs, then \textsc{$c$-Exact Set Cover} admits a kernel of size~$\Oh(k^{c -\epsilon})$, implying that \PHC~\cite{DM12,HW12}.
\end{proof}
\fi

\subsection{Connected Vertex Cover}

We now provide kernels for \ConVC{}, a well-studied variant of \textsc{Vertex Cover} which notoriously does not admit a polynomial kernel when parameterized~$k$~\cite{DLS14}.

\problemdef
{Connected Vertex Cover}
{A graph $G$ and $k \in \mathds{N}$.}
{Is there a set $S$ of at most $k$ vertices such that $S$ covers all edges of $G$ and $G[S]$ is connected?}

We will show that by applying~\Cref{thm:bound-sets-p-and-q} we obtain a kernel of size~$k^{\Oh(\gamma)}$. This kernel is complemented by a polynomial kernel for the parameter~$k+c$. Thus, \textsc{Connected Vertex Cover} is very different from \CapVC{} concerning the kernelization complexity on~$c$-closed graphs. \iflong
\subparagraph{A Polynomial Kernel in Weakly Closed Graphs.}
\fi
We may use the following known rule.  

\begin{rrule}[\cite{CGH17}]
  \label{rr:twinset}
  If $S \subseteq V(G)$ is a set of at least two twin vertices with a common neighborhood $N(S)$ such that $|S| > |N(S)|$, then remove one vertex~$v$ of $S$ from $G$.
\end{rrule}

After the exhaustive application of~\Cref{rr:twinset} we have, again by
\Cref{thm:bound-sets-p-and-q}, that the size of yes-instances is~$k^{\Oh(\gamma)}$, which
bounds the size of the independent set in terms of the vertex cover size and the size of a
largest twin set. The proof is completely analogous to the one of~\Cref{thm:capvc}.
\begin{theorem}
\label{thm-con-vc-kernel-gamma}
  \ConVC{}{} admits a kernel of size $k^{\Oh(\gamma)}$.
\end{theorem}

This kernel is essentially tight, because there is no kernel of size~$k^{o(d)}$~\cite{CGH17}.

\iflong
\subparagraph{A Polynomial Kernel for $k+c$.}

In order to obtain a polynomial kernel for \ConVC{}{} parameterized by~$k+c$, let us introduce an ``annotated'' version of \ConVC{}{} defined as follows.
In a \emph{red and white graph} $G$, the vertex set is partitioned into two sets: the red and white vertices, which we denote by $V_R(G)$ and by~$V_W(G)$, respectively.
The annotated version of \ConVC{}{} imposes the additional constraint that all the red vertices must be included into the solution.

\iflong
\problemdef
{Annotated Connected Vertex Cover}
{A red and white graph $G$ and $k \in \mathds{N}$.}
{Is there a set $S \supseteq V_R(G)$ of at most $k$ vertices such that $G[S]$ is connected and at least one endpoint of $e$ is included in $S$ for each edge $e \in E(G)$?}
\fi

The first reduction rule takes care of several trivial cases; the correctness is obvious.
\begin{rrule}
  \label{rr:isolated}
  Remove all isolated white vertices.

  If~$G$ has two connected components that contain
  edges then return No.

  If~$G$ has two connected components with red vertices, then return No. 

  If~$G$ has a solution of size one, then return Yes. 
  
\end{rrule}

We say that a vertex $v \in V(G)$ is \emph{simplicial} if its neighborhood forms a clique.
In particular, note that any degree-one vertex is simplicial.
We remove simplicial vertices in the next rule.

\begin{rrule}
  \label{rr:simplicial}
  If there is a simplicial vertex $v \in V(G)$, then we do as follows:
  \begin{itemize}
    \item
      If $v \in V_R(G)$, then decrease $k$ by 1.
    \item
      If $v \in V_W(G)$ or $\deg_G(v) = 1$, then color all the vertices in $N_G(v)$ red.
    \item
      Remove $v$.
  \end{itemize}
\end{rrule}

We assume that $|V(G)| \ge 3$ and that $G$ is connected.

\begin{lemma}
\label{lemma-simplical-vertices}
  \Cref{rr:simplicial} is correct.
\end{lemma}
\begin{proof}
  Let $v$ be a simplicial vertex in $G$ and let $(G', k')$ be the instance obtained by applying \Cref{rr:simplicial}.
  Suppose that $(G, k)$ is a yes-instance with a solution $S$.
  If $v \notin S$, then $v$ is white by definition and hence $k' = k$.
  Consequently, $S$ is a solution of $(G', k')$.
  So assume that $v \in S$.

  We first show that if $v$ is red, then $S' = S \setminus \{ v \}$ is a solution of $(G', k')$.
  This is clear for the case $\deg_G(v) > 1$.
  To see why, note that $v$ is simplicial in $G[S]$ and thus $S$ remains connected after deleting $v$.
  Suppose that $\deg_G(v) = 1$.
  Let $u$ be the neighbor of $v$.
  Since \Cref{rr:simplicial} colors $u$ red, we have to show that $u \in S$.
  Observe that $S \setminus \{ v \}$ contains at least one vertex by \Cref{rr:isolated}.
  Thus, $S$ must include $u$ for $G[S]$ to be connected.

  Now, consider the case that $v$ is white.
  Suppose that there exists a vertex $u \in N_G(v)$ with $u \notin S$.
  Since $N_G(v)$ is complete in~$G$, all the vertices of $N_G(v) \setminus \{ u \}$ are included in~$S$.
  It follows that $(S \setminus \{ v \}) \cup \{ u \}$ is also a solution of $(G, k)$.
  Hence, we can assume that~$S \supseteq N_G(v)$, thereby showing that $S \setminus \{ v \}$ is a solution of $(G', k')$.

  Conversely, suppose that $S'$ is a solution of $(G', k')$.
  If $v$ is white, $S'$ is also a solution of~$(G, k)$, because all the neighbors of $v$ are colored red by \Cref{rr:simplicial}.
  Otherwise, $S' \cup \{ v \}$ is a vertex cover of size $k' + 1 \le k$.
  Moreover, $S' \cup \{ v \}$ is connected in $G$ because~$S'$ includes a neighbor of $v$:
  If $\deg_G(v) = 1$, then the sole neighbor $u$ of $v$ is red in $G'$ due to \Cref{rr:simplicial}.
  Otherwise, $S'$ contains at least $|N_G(v)| - 1 \ge 1$ vertices of $N_G(v)$, because $N_G(v)$ forms a clique in $G'$.
\end{proof}

Now, we show that \Cref{rr:isolated,rr:simplicial} yield a polynomial kernel for \textsc{Annotated Connected Vertex Cover} in $c$-closed graphs.

\begin{lemma}
  \label{lemma:convckernel}
  {\normalfont\textsc{Annotated Connected Vertex Cover}} has a kernel with $\Oh(ck^2)$ vertices.
\end{lemma}
\begin{proof}
  We claim that an \textsc{Annotated Connected Vertex Cover} instance is a No-instance if it has at least $k + c \binom{k}{2}$ vertices after \Cref{rr:isolated,rr:simplicial} are exhaustively applied.
  Suppose that an instance $(G, k)$ of \textsc{Annotated Connected Vertex Cover} is a yes-instance with a solution $S$.
  Then, it holds for each vertex $v \in V(G) \setminus S$ that~$N_G(v) \subseteq S$ and that $\deg_G(v) \ge 1$ by \Cref{rr:isolated}.
  Moreover, each vertex $v \in V(G) \setminus S$ must have at least two nonadjacent neighbors in $S$ by \Cref{rr:simplicial}.
  It follows that $|V(G) \setminus S| < c \binom{|S|}{2} \le c \binom{k}{2}$ and thus $|V(G)| = |S| + |V(G) \setminus S| < k + c \binom{k}{2}$.
\end{proof}

Finally, we can reduce from \textsc{Annotated Connected Vertex Cover} to \textsc{Connected Vertex Cover} by attaching a leaf-vertex to each red vertex.
\fi

\begin{theorem}
  \label{thm:convc}
  \ConVC{} has a kernel with $\Oh(ck^2)$ vertices.
\end{theorem}

\iflong
\begin{proof}
  Let $(G, k)$ be an instance of \ConVC{}{}.
  We construct an \textsc{Annotated Connected Vertex Cover} instance $(H, k)$ where all the vertices in $G$ are colored white (that is, $V_R(H) = \emptyset$, $V_W(H) = V(G)$, and $E(H) = E(G)$).
  By \Cref{lemma:convckernel}, we obtain an instance of $(H', k')$ with $\Oh(ck^2)$ vertices in polynomial time.
  Let $G'$ be the graph obtained from $H'$ by attaching a degree-one vertex $\ell_v$ to each red vertex $v$.

  If there is a connected vertex cover in $H'$ including all the red vertices, then it is also a connected vertex cover in $G'$.
  Conversely, suppose that $S$ is a connected vertex cover of $G'$.
  Then, we can assume that~$S$ contains all red vertices of $H'$.
  So~$S$ is a solution of $(H', k')$.
\end{proof}
\fi

\subsection{Induced Matching}
\label{sec:im}

In this section, we provide a kernel of size $(\gamma k)^{\Oh(\gamma)}$ for \IM{}:

\problemdef
{Induced Matching}
{A graph $G$ and $k \in \mathbb{N}$.}
{Is there a set $M$ of at least $k$ edges such that the endpoints of distinct edges are pairwise nonadjacent?}

\IM{} is W[1]-hard for the parameter $k$ on general graphs.
For $c$-closed graphs, we~recently developed a kernel with $\Oh(c^7 k^8)$ vertices~\cite{KKS20}.
As for $d$-degenerate graphs, Kanj et al.~\cite{KPSX11} and Erman et al.~\cite{EKKW10} independently presented kernels of size $k^{\Oh(d)}$.
Later, Cygan et al.~\cite{CGH17} provided a matching lower bound $k^{o(d)}$  on the kernel size.
Note that this also implies the nonexistence of $k^{o(\gamma)}$-size kernels unless \PHC.

It turns out that \Cref{thm:bound-sets-p-and-q} is again helpful in designing a $k^{\Oh(\gamma)}$-size kernel for \IM{}.
In a nutshell, we show that the application of a series of reduction rules results in a graph with a $(\gamma k)^{\Oh(1)}$-size vertex cover.
We do so by combining different known kernelizations, namely, the one of Erman et al.~\cite{EKKW10} with the one in our previous work~\cite{KKS20}.
\Cref{thm:bound-sets-p-and-q} and the reduction rule which removes twin vertices then gives us a kernel of size~$(\gamma k)^{\Oh(\gamma)}$.

Erman et al.~\cite{EKKW10} use the following observation for~$d$-degenerate graphs.
\begin{lemma}[\cite{EKKW10}]
  \label{lemma:im:sparseim}
  Any graph $G$ with a matching $M$ has an induced matching of size $|M| / (4d_G + 1)$.
\end{lemma}

Ideally, we would like to prove a lemma analogous to \Cref{lemma:im:sparseim} on weakly $\gamma$-closed graphs.
Note, however, that a complete graph on $n$ vertices (which is weakly 1-closed) has no induced matching of size 2, although it contains a matching of size $\lfloor n / 2 \rfloor$.
So our algorithm needs to follow a slightly different route, and we prove an analogous lemma on weakly~$\gamma$-closed bipartite graphs (there exist bipartite $2$-closed graphs whose degeneracy is unbounded; see e.g. Eschen et al.~\cite{EHSS11}).
As we shall see, this serves our purposes.

\begin{lemma}
  \label{lemma:im:findim}
  Suppose that $G$ is a bipartite graph with a bipartition $(A, B)$.
  If $G$ has a matching $M$ of size $f_\gamma(k) := 4 \gamma k^2 + 3k$, then $G$ has an induced matching of size $k$.
\end{lemma}
\begin{proof}
Recall that~$\q^\sigma(v) := \{ u \in N(v) \mid \text{$u$ appears after $v$ in $\sigma$} \}$.
  Let $S \subseteq V(G)$ be the set of vertices $v$ such that $|Q(v)| \ge \gamma k$.
  Suppose that $|S| \ge 2k$.
  Then, we may assume that~$|A \cap S| \ge k$.
  Let $A' \subseteq A \cap S$ be an arbitrary vertex set of size exactly $k$ and consider some vertex $v \in A'$.
  Since $|Q(v) \cap N(v')| < \gamma$ for every $v' \in A' \setminus \{ v \}$, we have~$|Q(v) \setminus \bigcup_{v' \in A' \setminus \{ v \}} N(v')| > 0$ for each~$v\in A'$. 
  Consequently, there is at least one vertex~$q_v\in Q(v) \setminus \bigcup_{v' \in A' \setminus \{ v \}} N(v')$.
  Then, the edge set $\{ v q_v \mid v \in A' \}$ forms an induced matching of size $k$ in $G$.

  Now, consider the case $|S| < 2k$.
  By the definition of $S$, it holds that $|Q_{G - S}(v)| \le |Q_G(v)| \le \gamma k$ for each vertex $v \in V(G) \setminus S$.
   Hence, the degeneracy of $G - S$ is at most~$\gamma k$.
  Since~$G-S$ has a matching $M_{G - S}$ of size at least~$|M| - |S| \ge 4 \gamma k^2 + k$, \Cref{lemma:im:sparseim} yields an induced matching of size $|M_{G - S}| / (4 d_{G - S} + 1) \ge k$.
\end{proof}

We use the following reduction rule to sparsify the graph~$G$ such that each sufficiently large vertex set contains a large independent set (see \Cref{lemma:im:findis}).

\begin{rrule}
  \label{rr:im:delv}
  If for some vertex $v \in V(G)$, there is a maximum matching $M_v$ of size at least $2 \gamma k$ in $G[Q(v)]$, then delete $v$.
\end{rrule}

\begin{lemma}
  \Cref{rr:im:delv} is correct.
\end{lemma}
\begin{proof}
  Let $G' := G - v$.
  Suppose that $G$ has an induced matching $M$ of size $k$.
  If~$v \notin V(M)$, then $M$ is also an induced matching in $G'$.
  So assume that $vv' \in V(M)$ for some vertex~$v' \in V(G)$.
  Then, we have $|N(u) \cap Q(v)| < \gamma$ for any vertex $u \in V(M  \setminus \{ vv' \})$ and thus $|V(M_v) \setminus \bigcup_{u \in V(M \setminus \{ vv'\})} N(u)| \ge 2|M_v| - (\gamma - 1) (2k - 2) > |M_v|$.
  By the pigeon-hole principle, this implies that there is an edge $e \in M_v$ not incident to any vertex in $V(M)$ and no endpoint of~$e$ is adjacent to any vertex in~$V(M\setminus\{vv'\})$. 
  Then, $(M \setminus \{ vv' \}) \cup \{ e \}$ is an induced matching of size $k$ in $G'$.
\end{proof}

\begin{lemma}
  \label{lemma:im:findis}
  Suppose that $G$ is a graph in which \Cref{rr:im:delv} is applied on every vertex.
  Then, every vertex set~$S \subseteq V(G)$ of size at least $g_\gamma(k) := 4 \gamma k^2 + k^2$  contains an independent set $I \subseteq S$ of size $k$.
\end{lemma}
\begin{proof}
  Suppose that there is no independent set of size $k$ in $G' := G[S]$ for some vertex set~$S$ of size~$g_\gamma(k)$.
  For every vertex~$v \in S$, let $M_v$ be a maximum matching in $Q_{G'}(v)$ and let~$I_v := Q_{G'}(v) \setminus V(M_v)$.
  By \Cref{rr:im:delv}, we have $|V(M_v)| = 2 |M_v| \le 4\gamma k$.
  Since~$I_v$ is an independent set, we then have $|Q_{G'}(v)| = |M_v| + |I_v| < 4 \gamma k + k$ for every vertex~$v \in S$, and thus $d_{G'} < 4 \gamma k + k$.
  Note, however, that $G'$ has an independent set of size~$|S| / (d_{G'} + 1) \ge k$, which is a contradiction.
\end{proof}

To identify a part of the graph with a sufficiently large induced matching, we rely on the LP relaxation of \textsc{Vertex Cover}, following our previous approach to obtain a polynomial kernel on $c$-closed graphs~\cite{KKS20}.
Recall that \textsc{Vertex Cover} can be formulated as an integer linear program as follows, using a variable $x_v$ for each $v \in V(G)$:
\begin{align*}
    \min \sum_{v \in V(G)} x_v
    \qquad\text{subject to} \quad & x_u + x_v \ge 1 \quad \forall uv \in E(G), \\[-2ex]
                                  & x_v \in \{ 0, 1 \} \quad \forall v \in V(G).
\end{align*}
We will refer to the LP relaxation of \textsc{Vertex Cover} as VCLP.
We use the well-known facts that VCLP always admits an optimal solution in which $x_v \in \{ 0, 1/2, 1 \}$ for each $v \in V(G)$ and that such a solution can be found in polynomial time.
Suppose that we have such an optimal solution $(x_v)_{v \in V(G)}$.
Let $V_0 := \{ v \in V(G) \mid x_v = 0 \}$, $V_1 := \{ v \in V(G) \mid x_v = 1 \}$, and $V_{1/2} := \{ v \in V(G) \mid x_v = 1 / 2 \}$.
Also, let $\opt(G)$ be the optimum of VCLP.
We show that we can immediately return Yes, whenever $\opt(G)$ is sufficiently large:

\begin{rrule}
  \label{rr:im:vpos}
  If $\opt(G) \ge 2 g_\gamma(g_\gamma(f_\gamma(k)))$, then return Yes.
\end{rrule}

Here, the functions $f_\gamma$ and $g_\gamma$ are as specified in \Cref{lemma:im:findim,lemma:im:findis}, respectively.

\begin{lemma}
  \Cref{rr:im:vpos} is correct.
\end{lemma}
\begin{proof}
  We show that $G$ has an induced matching of size $k$ whenever $\opt(G) \ge 2 g_\gamma(g_\gamma(f_\gamma(k)))$.
  Let $M$ be an arbitrary maximal matching in $G$.
  Since $V(M)$ is a vertex cover, we have~$\opt(G) \le |V(M)| = 2|M|$, and hence $|M| \ge \opt(G) / 2 \ge g_\gamma(g_\gamma(f_\gamma(k)))$.
  Let~$M := \{ a_1 b_1, \dots, a_{|M|} b_{|M|} \}$ and let $A := \{ a_1, \dots, a_{|M|} \}$ and $B := \{ b_1, \dots, b_{|M|} \}$.
  By \Cref{lemma:im:findis}, there exists an independent set $A' \subseteq A$ of size $s' := g_{\gamma}(f_\gamma(k))$.
  Without loss of generality, suppose that $A' = \{ a_1, \dots, a_{s'} \}$ and let $B' := \{ b_1, \dots, b_{s'} \}$ be the set of vertices matched to~$A'$ in $M$.
  Again by \Cref{lemma:im:findis}, we obtain an independent set $B'' \subseteq B'$ of size $s'' := f_\gamma(k)$.
  We assume without loss of generality that $B'' = \{ b_1, \dots, b_{s''} \}$.
  Let $A'' := \{ a_1, \dots, a_{s''} \}$ be the set of vertices matched to $B''$ in $M$.
  Then, $G[A'' \cup B'']$ is a bipartite graph with a matching of size at least $s'' = f_\gamma(k)$.
  Consequently, by \Cref{lemma:im:findim}, there is an induced matching of size $k$ in $G[A'' \cup B'']$.
\end{proof}

Since $\opt(G) = |V_{1/2}| / 2 + |V_1|$, it holds that $|V_{1/2}| / 2 + |V_1| \le 2 g_\gamma(g_\gamma(f_\gamma(k))) \in \Oh(\gamma^7 k^8)$ after the application of \Cref{rr:im:vpos}.
Hence, it remains to bound the size of $V_0$.
To do so, it suffices to simply remove ``twins'':

\begin{rrule}
  If $N(u) = N(v)$ for some vertices $u, v \in V(G)$, then delete $v$.
\end{rrule}

Since an induced matching contains at most one of~$u$ and~$v$, the rule is obviously correct.  
We are finally ready to utilize \Cref{thm:bound-sets-p-and-q} to derive an upper bound on $V_0$:
Since $V_0$ is an independent set, \Cref{thm:bound-sets-p-and-q} gives us $|V_0| \in |V_{1/2} \cup V_{1}|^{\Oh(\gamma)} \in (\gamma k)^{\Oh(\gamma)}$.
Thus, we have the following:

\begin{theorem}
  \IM{} has a kernel of size $(\gamma k)^{\Oh(\gamma)}$.
\end{theorem}

\section{Independent Set and Ramsey-Type Problems}

In this section, we further investigate the kernelization complexity of \textsc{Independent Set}:

\problemdef{Independent Set}
{A graph $G$ and $k \in \mathds{N}$.}
{Is there a set $I$ of at least $k$ vertices that are pairwise nonadjacent?}

In previous work, we provided kernels with $\Oh(ck^2)$~\cite{KKS20} (and $\Oh(\gamma k^2)$~\cite{KKS20a}) vertices whose maximum degree (degeneracy, respectively) is at most $ck$ ($\gamma k$, respectively). This results in kernels of size $\Oh(c^2 k^3)$ \cite{KKS20} and $\Oh(\gamma^2 k^3)$ \cite{KKS20a}, respectively.
We complement these results by showing that, unless \PHC{}, \IS{} admits no kernel of size $k^{2-\epsilon}$ and no kernel with~$k^{4/3 - \varepsilon}$ vertices even if the~$c$-closure is constant.
We also show a kernel lower bound of size~$c^{1 - \varepsilon} k^{\Oh(1)}$ for any~$\varepsilon$.

We then turn our attention to the following related problem where $\mathcal{G}$ is a graph class containing all complete graphs and edgeless graphs.

\problemdef{$\mathcal{G}$-Subgraph}
{A graph $G$ and $k \in \mathds{N}$.}
{Is there a set $S$ of at least $k$ vertices such that $G[S] \in \mathcal{G}$?}

Khot and Raman~\cite{KR02} showed that \textsc{$\mathcal{G}$-Subgraph} is FPT when parameterized by $k$, using Ramsey's theorem,
which states that for any $k \in \mathds{N}$, there exists some number $R(k)\in 2^{\Oh(k)}$ such that any graph $G$ on at least $R(k)$ vertices contains a clique of size $k$ or an independent set of size $k$.
\Ramsey{} is the special case when~$\mathcal{G}$ consists of all complete and edgeless graphs and admits no polynomial kernel unless \PHC{}~\cite{Kra14}.
  Similarly, \textsc{$\mathcal{G}$-Subgraph} admits no polynomial kernel for several graphs classes $\mathcal{G}$, such as cluster graphs~\cite{KPRR14}.
Our contribution for \textsc{$\mathcal{G}$-Subgraph} is two-fold: First, we observe that the lower bounds for \textsc{Independent Set} on graphs with constant $c$-closure also hold for \Ramsey{}. This complements a kernel for \textsc{$\mathcal{G}$-Subgraph} with $\Oh(c k^2)$ vertices that we showed in previous work \cite{KKS20}.\footnote{We showed that any $c$-closed $n$-vertex graph  contains a clique or an independent set of size~$\Omega(\sqrt{n / c})$.}
Second, we provide a kernel of size~$k^{\Oh(\gamma)}$.
To show our kernel lower bounds, we will use \emph{weak~$q$-compositions}.
Weak~$q$-compositions exclude kernels of size $\Oh(k^{q - \varepsilon})$ for $\varepsilon > 0$.

\begin{definition}[\cite{DM12,HW12}]
  Let~$q\ge 1$ be an integer, let~$L_1\subseteq \{0,1\}^*$ be a classic (non-parameterized) problem, and let~$L_2\subseteq \{0,1\}^*\times\mathds{N}$ be a parameterized  problem.  A \emph{weak~$q$-composition} from~$L_1$ to~$L_2$ is a polynomial-time algorithm that on input~$x_1, \ldots ,x_{t^q}\in\{0,1\}^n$ outputs an instance~$(y,k')\in\{0,1\}^*\times\mathds{N}$ such that:
\iflong  \begin{itemize}
    \item $(y,k')\in L_2 \Leftrightarrow x_i\in L_1 \text{ for some } i \in [t^q]$, and
    \item $k'\le t\cdot n^{\Oh(1)}$.
  \end{itemize}
\fi
\end{definition}

\begin{lemma}[\cite{CFK+15,DM12,HW12}] \label{lemma:is:noq}
  Let~$q\ge 1$ be an integer, let~$L_1\subseteq \{0,1\}^*$ be a classic NP-hard problem, and let~$L_2\subseteq \{0,1\}^*\times\mathds{N}$ be a parameterized problem.
  The existence of a weak~$q$-composition from~$L_1$ to~$L_2$ implies that~$L_2$ has no compression of size~$\Oh(k^{q-\epsilon})$ for any~$\epsilon >0$, unless \PHC.
\end{lemma}

\begin{figure}[t]
  \begin{center}
    \begin{tikzpicture}
      \node[vertex, label={$p^i_{r, 1, 1}$}] (u1) at (0, 0) {};
      \node[vertex, label={$p^i_{r, 2, 2}$}] (u2) at (1, 0) {};
      \node[vertex, label={$p^i_{r, 2, 1}$}] (u3) at (2, 0) {};
      \node[vertex, label={$p^i_{r, 3, 2}$}] (u4) at (3, 0) {};
      \node[vertex, label={$p^i_{r, 3, 1}$}] (u5) at (4, 0) {};
      \node[vertex, label={$p^i_{r, 4, 2}$}] (u6) at (5, 0) {};
      \node[vertex, label={$p^i_{r, 4, 1}$}] (u7) at (6, 0) {};
      \node[vertex, label={$p^i_{r, 5, 2}$}] (u8) at (7, 0) {};

      \node[vertex, label=below:$p^{i + 1}_{r, 1, 1}$] (v1) at (0, -1) {};
      \node[vertex, label=below:$p^{i + 1}_{r, 2, 2}$] (v2) at (1, -1) {};
      \node[vertex, label=below:$p^{i + 1}_{r, 2, 1}$] (v3) at (2, -1) {};
      \node[vertex, label=below:$p^{i + 1}_{r, 3, 2}$] (v4) at (3, -1) {};
      \node[vertex, label=below:$p^{i + 1}_{r, 3, 1}$] (v5) at (4, -1) {};
      \node[vertex, label=below:$p^{i + 1}_{r, 4, 2}$] (v6) at (5, -1) {};
      \node[vertex, label=below:$p^{i + 1}_{r, 4, 1}$] (v7) at (6, -1) {};
      \node[vertex, label=below:$p^{i + 1}_{r, 5, 2}$] (v8) at (7, -1) {};

      \draw (u1) -- (u2) -- (u3) -- (u4) -- (u5) -- (u6) -- (u7) -- (u8);
      \draw (v1) -- (v2) -- (v3) -- (v4) -- (v5) -- (v6) -- (v7) -- (v8);

      \draw (u1) -- (v2); \draw (u2) -- (v1);
      \draw (u3) -- (v4); \draw (u4) -- (v3);
      \draw (u5) -- (v6); \draw (u6) -- (v5);
      \draw (u7) -- (v8); \draw (u8) -- (v7);
    \end{tikzpicture}
  \end{center}
  \caption{An illustration of $P^i_r$ and $P_{i + 1}^r$ for $t = 5$.}
  \label{fig:is:pir}
\end{figure}
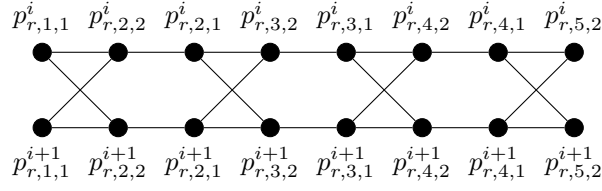

\subparagraph*{Weak composition.}
We provide a weak composition from the following problem:
\problemdef{Multicolored Independent Set}
{A graph $G$ and a partition $(V_1, \dots, V_k)$ of $V(G)$ into $k$ cliques.}
{Is there an independent set of size exactly $k$?}

A standard reduction from a restricted variant of \textsc{3-SAT} (for instance, each literal appears exactly twice \cite{BKS03})  shows that \textsc{Multicolored Independent Set} is NP-hard even when $\Delta_G \in \mathcal{O}(1)$ and $|V_i| \in \mathcal{O}(1)$ for all~$i \in [k]$.
Let $[t]^q$ be the set of $q$-dimensional vectors whose entries are in $[t]$.
Suppose now that~$q\ge 2$ is a constant and that we are given $t^q$ instances $\mathcal{I}_{x} = (G_{x}, (V_x^1, \dots, V^k_x))$ for $x \in [t]^q$, where $\Delta_{G_x} \in \Oh(1)$ and $V^i_x \in \Oh(1)$ for all $x \in [t]^q$ and $i \in [k]$. 
We will construct an \IS{} instance $(H, k')$.
 The kernel lower bound of size~$k^{2-\epsilon}$ will be based on the special case where~$q=2$. To obtain the lower bound of~$c^{1-\epsilon}k^{\Oh(1)}$ we need, however, that the composition works for all~$q\in \mathds{N}$. Hence, we give a generic description in the following.
First, we construct a graph $H_i$ as follows for every $i \in [k]$:
\begin{itemize}
  \item
    For every $x \in [t]^q$, include $V^i_x$ into $V(H_i)$.
  \item
    For every $r \in [q]$, introduce a path $P^i_r$ on $2t - 2$ vertices.
    We label the $(2j - 1)$-th vertex as $p^{i}_{r, j, 1}$ and the $2j$-th vertex as $p^{i}_{r, j + 1, 2}$ (see \Cref{fig:is:pir} for an illustration).
    Note that~$V(P^i_r) = \{ p^i_{r, j, 1}, p^i_{r, j + 1, 2} \mid j \in [t - 1] \}$.
    For every $j \in [t]$, we define $P^i_{r, j}$: let~$P^{i}_{r, 1} = \{ p^i_{r, 1, 1} \}$, $P^i_{r, t} = \{ p^i_{r, t, 2} \}$, and $P^i_{r, j} = \{ p^i_{r, j, 1}, p^i_{r, j, 2} \}$ for $j \in [2, t - 2]$.
  \item
    For every $r \in [q]$ and $j \in [t]$, add edges such that $P^i_{r, j} \cup \bigcup_{x \in [t]^q, x_r = j} V_x^i$ forms a clique.
\end{itemize}

We then construct the graph $H$.
We start with a disjoint union of $H_i$ for $i \in [k]$ and add the following edges:
\begin{itemize}
  \item
    For every $x \in [t]^q$, we add edges such that $H[V(G_x)] = G_x$.
  \item
    For every $i \in [k - 1]$, $r \in [q]$, and $j \in [t - 1]$, add edges $p^i_{r, j, 1} p^{i + 1}_{r, j + 1, 2}$ and $p^{i + 1}_{r, j, 1} p^{i}_{r, j + 1, 2}$.
\end{itemize}

This concludes the construction of $H$.
Let $k' := qkt - qk + k$.

We will call the vertices of $\bigcup_{x \in [t]^q, i \in [k]} V^i_x$ the \emph{instance vertices}.
The other vertices, which are on $P^i_r$ for some $i \in [k]$ and~$r \in [q]$, serve as \emph{instance selectors}:
As we shall see later, any independent set $J$ of size $k'$ in $H$ contains exactly $t - 1$ vertices of $P^i_r$ for every $i \in [k]$ and~$r \in [q]$.
In fact, there is exactly one $j \in [t]$ such that $J \cap P^i_{r, j} = \emptyset$ and $|J \cap P^i_{r, j'}| = 1$ for all $j' \in [t] \setminus \{ j \}$.
Consequently, $J$ does not contain any instance vertex in $V^i_x$ for $x_r \ne j$, and thereby, $j$ is \emph{selected} for the $r$-th dimension.

We examine the $c$-closure of the constructed graph $H$ and prove the correctness.

\begin{lemma}
  \label{lemma:is:ch}
  It holds that $\cl(H) \in \Oh(t^{q - 2})$.
\end{lemma}
\iflong
\begin{proof}
  We first show that $\cl_H(v^i_x) \in \Oh(t^{q - 2})$ for every instance vertex $v^i_x \in V^i_x$.
  More precisely, we show that $|N(v) \cap N(v^i_x)| \in \Oh(t^{q - 2})$ for every vertex $v \in V(H) \setminus N(v^i_x)$.
  By construction, it holds that
  \begin{align*}
    N(v^i_x) \subseteq U_i^x \cup V(G_x) \cup \bigcup_{r \in [q]} P^i_{r, x_r},\, \text{where } U_i^x = \bigcup_{x' \in \mathcal{X}^q_t, \, \exists r \in [q] \colon x_r = x_r'} V_i^{x'}.
  \end{align*}
  
  Since $|P^i_{r, x_r}| \le 2$ for each $r \in [q]$, we have $|N(v) \cap \bigcup_{r \in [q]} P^i_{r, x_r}| \le 2q \in \Oh(1)$.
  Moreover, we have $|N(v) \cap V(G_x)| \in \Oh(1)$:
  If $v \in V(G_x)$, then $|N(v) \cap V(G_x)| \le \Delta_{G_x} \in \Oh(1)$.
  Otherwise, $N(v) \cap V(G_x) = V_{i'}^x$ for some $i' \in [k]$, and hence $|N(v) \cap V(G_x)| = |V_{i'}^x| \in \Oh(1)$.

  Thus, it suffices to show that either $v^i_x \in N(v)$ or $|N(v) \cap U_i^x| \in \Oh(t^{q - 2})$ holds for every vertex $v \in V(H)$.
  If $v \in V(H_{i'})$ for $i' \in [k] \setminus \{ i \}$, then $v$ has $\Oh(1)$ neighbors in $U_i^x$.
  Otherwise,~$v \in V_i^r$ or $v \in P^i_{r, y_r}$ for some $y \in [t]^q$.
  \begin{itemize}
    \item
      If $x_r = y_r$ for some $r \in [q]$, then $v^i_x$ and $v$ are adjacent.
    \item
      If $x_r \ne y_r$ for all $r \in [q]$, then $N(v) \cap U_i^x \subseteq \bigcup_{z \in \mathcal{Z}} V_i^z$, where $\mathcal{Z}$ is the set of vectors~$z \in [t]^q$ such that $x_{r'} = z_{r'}$ and $y_{r''} = z_{r''}$ for some $r', r'' \in [q]$.
      We claim that $\mathcal{Z}$ contains $\Oh(t^{q - 2})$~vectors.
      As $x_r \ne y_r$ for all $r \in [q]$, we can assume that $r' \ne r''$.
      So deleting the~$r$th and the~$r'$th entries from any vector in $\mathcal{Z}$ yields a vector in $[t]^{q - 2}$.
      Thus, there is a one-to-one correspondence between an element of $\mathcal{Z}$ and $r' \ne r'' \in [q]$, and a vector~$z' \in [t]^{q - 2}$.
      It follows that $|\mathcal{Z}| \le q(q - 1) t^{q - 2} \in \Oh(t^{q - 2})$, and consequently, $|N(v) \cap U_i^x| \le \Oh(|\mathcal{Z}|) = \Oh(t^{q - 2})$.
  \end{itemize}

  It remains to show that for every pair $v, v'$ of instance selector vertices, we have either~$vv'\in E(H)$ or $|N(v) \cap N(v')| \in \Oh(t^{q - 2})$.
  Suppose that $v \in P^i_{j, r}$ and $v \in P^{i'}_{j', r'}$.
  If $i \ne i'$, then~$v$ and $v'$ have a constant number of neighbors in common.
  So assume that $i = i'$.
  If~$r = r'$, then either $vv' \in E(H)$ or $N(v) \cap N(v') = \emptyset$ holds.
  Otherwise, it holds that $N(v) \cap N(v') \subseteq \bigcup_{z \in \mathcal{Z}} V_i^z$ where $\mathcal{Z}$ is the set of vectors $z$ such that $z_r = j$ and $z_{r'} = j'$.
  It follows from an analogous argument as above that $|\mathcal{Z}| = t^{q - 2}$.
  Since $|V^i_x| \in \Oh(1)$ for every~$i \in [k]$ and $x \in [t]^q$, we have $|N(v) \cap N(v')| \in \Oh(t^{q - 2})$.
\end{proof}
\fi

\iflong

\begin{lemma}
  \label{lemma:is:cor1}
  If $G_x$ has a multicolored independent set $I$ of size $k$ for some $x \in [t]^q$, then $H$ has an independent set of size $k'$.
\end{lemma}
\begin{proof}
  One can easily verify that the set $J$ consisting of the following vertices is an independent set in $H$.
  \begin{itemize}
    \item All vertices of $I$.
    \item The vertex $p^i_{r, j, 1}$ for each $r \in [q]$, $i \in [k]$, and $j \in [x_r - 1]$.
    \item The vertex $p^i_{r, j, 2}$ for each $r \in [q]$, $i \in [k]$, and $j \in [x_r + 1, t]$.
  \end{itemize}
  The size of $J$ is then $k + \sum_{ r \in [q], i \in [k]} (x_r - 1) + (t - x_r) = k + qk(t - 1) = k'$.
\end{proof}

\begin{lemma}
  \label{lemma:is:cor2}
  If $H$ has an independent set $J$ of size $k'$, then $G_x$ has a multicolored independent set of size $k$ for some $x \in [t]^q$.
\end{lemma}
\begin{proof}
  First, we show that $|J \cap \bigcup_{x \in [t]^q} V^i_x| = 1$ and $|J \cap V(P^i_r)| = t - 1$ for every $i \in [k]$ and $r \in [q]$.
  Let $J^i := J \cap V(H_i)$ for each $i \in [k]$.
  Since $P^i_r$ is a path on $2t - 2$ vertices, $J^i$ contains at most $t - 1$ vertices of $P^i_r$ for every $r \in [q - 1]$, that is, $|J^i \cap V(P^i_r)| \le t - 1$.
  The set of remaining vertices in $H_i$, that is, $V(H_i) \setminus \bigcup_{r \in [q - 1]} V(P^i_r)$, can be partitioned into $t$ cliques:
  $P^i_{q, j} \cup \bigcup_{x \in [t]^q, x_q = j} V_x^i$ for $j \in [t]$.
  So we have $|J^i \cap (P^i_{q, j} \cup \bigcup_{x \in [t]^q, x_q = j} V_x^i)| \le 1$ for each~$j\in[t]$.
  Consequently, we obtain
  \begin{align*}
    |J|
    &= \sum_{i \in [k], r \in [q - 1]} |J^i \cap P^i_r| + \sum_{i \in [k], j \in [t]} |J^i \cap (P^i_{q, j} \cup \bigcup_{x \in [t]^q, x_q = j} V_x^i)|\\
    &\le k(q - 1)(t - 1) + kt = qkt - qk + k = k'.
  \end{align*}
  
  In fact, the equality holds.
  Hence, $|J^i \cap P^i_r| = t - 1$ and $|J^i \cap (P^i_{q, j} \cup \bigcup_{x \in [t]^q, x_q = j} V_x^i)| = 1$ for each $i \in [k]$ and $r \in [q - 1]$.
  An analogous argument (in which $J^i$ is partitioned into $J^i \cap P^i_r$ for $r \in [2, q]$ and $J^i \cap (P^i_{1,j} \bigcup_{x \in [t]^q, x_r = j} V_x^i)$) also shows that $|J^i \cap P^i_q| = t - 1$.
  We thus have for each $i \in [k]$,
  \begin{align*}
    |J^i \cap \bigcup_{x \in [t]^q, x_r = j} V_x^i| = |J^i \cap (\bigcup_{j \in [t]} (P^i_{q, j} \cup \bigcup_{x \in [t]^q, x_r = j} V_x^i))| - |J^i \cap P^i_q| = t - (t - 1) = 1.
  \end{align*}

  Now, let $x^i \in [t]^q$ be such that $J \cap V_i^{x^i} \ne \emptyset$ for each $i \in [k]$.
  We claim that $J \cap V(P^i_r) = \{ p^i_{r, j, 1} \mid j \in [x_r^i - 1] \} \cup \{ p^i_{r, j, 2} \mid j \in [x_r^i + 1, t]\}$ for every $i \in [k]$ and $r \in [q]$:
  \begin{itemize}
    \item
      If $x_r^i = 1$, then $J \cap V(P^i_r)$ lies inside the path $P^i_r - \{ p^i_{r, 1, 1} \}$ on $2t - 3$ vertices, in which there is exactly one independent set of size $t - 1$, namely, $\{ p^i_{r, j, 2} \mid j \in [2, t] \}$.
      It follows that $J \cap V(P^i_r) = \{ p^i_{r, j, 2} \mid j \in [2, t] \}$.
    \item
      If $x_r^i = t$, then $J \cap V(P^i_r)$ lies inside the path $P^i_r - \{ p^i_{r, t, 2} \}$ on $2t - 3$ vertices, in which there is exactly one independent set of size $t - 1$, namely, $\{ p^i_{r, j, 1} \mid j \in [t - 1] \}$.
      It follows that $J \cap V(P^i_r) = \{ p^i_{r, j, 1} \mid j \in [t - 1] \}$.
    \item
      If $x_r^i \in [2, t - 1]$, then $J \cap V(P^i_r)$ lies in the disjoint union of $P^i_r[\bigcup_{j \in [x_r^i - 1]} P^i_{r, j}]$ and $P^i_r[\bigcup_{j \in [x_r^i + 1, t]} P^i_{r, j}]$.
      Since $P^i_r[\bigcup_{j \in [x_r^i - 1]} P^i_{r, j}]$ (and $P^i_r[\bigcup_{j \in [x_r^i + 1, t]} P^i_{r, j}]$) is a path on $2x_r^i - 3$ (and $2t - 2x_r^i - 1$, respectively) vertices, it has only one independent set of size $x_r^i - 1$ (and $t - x_r^i$, respectively): $\{ p^i_{r, j, 1} \mid j \in [x_r - 1] \}$ (and $\{ p^i_{r, j, 2} \mid j \in [x_r + 1, t]\}$, respectively)
      Thus, our claim holds.
  \end{itemize}

  Finally, we show that $x^i = x^{i + 1}$ for all $i \in [k - 1]$.
  Assume to the contrary that $x_r^i \ne x^{i + 1}_r$ for some $i \in [k - 1]$ and $r \in [q]$.
  If $x_r^i < x_r^{i + 1}$ ($x_r^i > x_{r}^{i + 1}$), then $J$ contains $p^i_{r, x_r^i + 1, 2}$ and~$p^{i + 1}_{r, x_r^i, 1}$ ($p^i_{r, x_{r}^{i + 1}, 1}$ and $p^{i + 1}_{r, x_{r}^{i + 1} + 1, 2}$, respectively).
  By construction, these vertices are adjacent in $H$, which contradicts the fact that $J$ is an independent set.
  We thus have shown that $x^i = x^{i + 1}$ for all $i \in [k - 1]$, and thereby, there exists $x \in [t]^q$ such that $J \cap V^i_x \ne \emptyset$ for all $i \in [k]$.
  Since $H[\bigcup_{i \in [k]} V^i_x] = G_x$, it follows that $J \cap \bigcup_{i \in [k]} V^i_x$ is a multicolored independent set of size $k$ in $G_x$.
\end{proof}
\fi

For $q = 2$, we have a weak $2$-composition from \textsc{Multicolored Independent Set} to \IS{} on $\Oh(t^{q - 2}) = \Oh(1)$-closed graphs by \iflong\Cref{lemma:is:ch,lemma:is:cor1,lemma:is:cor2} \fi.
Since the constructed graph $H$ has no clique of size $k'$, the construction also constitutes a weak $2$-composition to \Ramsey{} on $\Oh(1)$-closed graphs.
Thus, \Cref{lemma:is:noq} implies the following:

\begin{theorem}
  \label{prop:is:lb}
  For any $\varepsilon > 0$, neither \IS{} nor \Ramsey{} has a kernel of size $k^{2 - \varepsilon}$ on graphs of constant $c$-closure, unless \PHC{}.
\end{theorem}

\Cref{prop:is:lb} immediately implies that neither \IS{} nor \Ramsey{} admits a kernel of $k^{1 - \varepsilon}$ vertices.
We improve this lower bound on the number of vertices, taking advantage of the fact that any $n$-vertex $c$-closed graph can be encoded in polynomial time in $\Oh(c n^{1.5} \log n)$~bits~\cite{EHSS11}.
Assume for a contradiction that \IS{} or \Ramsey{} admits a kernel of $k^{4/3 - \varepsilon'}$ vertices for constant $c$.
Using the above-mentioned encoding, we obtain a string with $\Oh(k^{(4/3 - \varepsilon')1.5}\log k) = \Oh( k^{2 - \varepsilon})$ bits.
So a kernel of $k^{4/3 - \varepsilon'}$~vertices implies that there is a compression of~\IS{} or~\Ramsey{} with bitsize $\Oh(k^{2 - \varepsilon})$, a contradiction.
Thus, we have the following:

\begin{theorem}
  For any $\varepsilon > 0$, neither \IS{} nor \Ramsey{} has a kernel of $k^{4/3 - \varepsilon}$~vertices on graphs of constant $c$-closure, unless \PHC{}.
\end{theorem}

We also obtain another kernel lower bound on \IS{};
this bound excludes the existence of polynomial kernels (in terms of $c + k$) whose dependence on $c$ is sublinear.

\begin{theorem}
  For any $\varepsilon > 0$, \IS{} has no kernel of size $c^{1 - \varepsilon} k^{\Oh(1)}$ unless \PHC{}.
\end{theorem}
\begin{proof}
  We show that \IS{} admits no kernel of size $c^{1 - \varepsilon} k^{i}$ for any $\varepsilon, i > 0$, unless~\PHC{}.
  Let $q$ be a sufficiently large integer with $\frac{q - \varepsilon - i}{q - 2} > 1 - \varepsilon$ (that is, $q > \frac{i + 3 \varepsilon - 2}{\varepsilon}$).
  A straightforward calculation shows that $\ell := c^{\frac{1}{q - 2} \left( 1 - \frac{i}{q - \varepsilon} \right)} k'^{\frac{i}{q - \varepsilon}} \in \Oh(t)$, and hence \IS{} admits a weak $q$-decomposition for the parameterization $\ell$.
  Thus, \Cref{lemma:is:noq} implies that there is no kernel of size $\ell^{q - \varepsilon} = c^{\frac{q - \varepsilon - i}{q - 2}} k'^i > c^{1 - \varepsilon} k'^i$.
\end{proof}

Finally, we show that \textsc{$\mathcal{G}$-Subgraph} has a kernel of size $k^{\Oh(\gamma)}$ \iflong for any graph class $\mathcal{G}$ containing all complete graphs and edgeless graphs\fi .

\begin{proposition}
  \label{thm:ramsey}
  Any graph $G$ on at least $R_\gamma(a, b) \in (a + b)^{\gamma + \Oh(1)}$ vertices has a clique of size $a$ or an independent set of size $b$.
\end{proposition}
\begin{proof}
  Ramsey's theorem states that any graph on at least $R(a, b)$ vertices contains a clique of size $a$ or an independent set of size $b$.
  It is known that $R(a, b) \le \binom{a + b}{b}$.
  So whenever $a \le \gamma$ or $b \le \gamma$, we have $R_{\gamma}(a, b) \le \binom{a + b}{\gamma} \in \Oh((a + b)^{\gamma})$.
  For $a, b > \gamma$, let $R_\gamma(a, b)$ be some number greater than $a \binom{b}{\gamma} + b \binom{a}{\gamma} \sum_{b' \in [b]} \binom{b' - 1}{\gamma}$.

  Let $n = |V(G)|$ and let $v_1, \dots, v_n$ be a closure ordering $\sigma$ of $G$.
  We divide $V(G)$ into $b$~subsets $V_1, \dots, V_b$ of equal size:
  let $V_i = \{ v_{(b - i)n / b + 1}, \dots, v_{(b - i + 1)n / b} \}$ for each $i \in [b]$.
  Notably, $V_1$ is the set of $n / b$ vertices occurring last in $\sigma$ and $V_b$ is the set of $n / b$ vertices occurring first in $\sigma$.
  Moreover, let $G_i := G[\{ v_{(b - i)n / b + 1}, \dots, v_n \}]$ be the induced subgraph by~$\bigcup_{i' \in [i]} V_{i'}$ for each $i \in [b]$.
  Suppose that $G$ contains no clique of size $a$.
  We will show that there is an independent set of size $b$ in $G$.
  More precisely, we prove by induction that $G_i$ contains an independent set of size $i$ for each $i \in [b]$.

  This clearly holds for $i = 1$.
  For $i > 1$, assume that there is an independent set~$I$ of size $i - 1$ in $G_{i - 1}$ by the induction hypothesis.
  Consider vertex sets $X \subseteq I$ of size $\gamma$ and~$V_X = \{ v \in V_i \mid N_G(v) \supseteq X \}$.
  Since $G$ is weakly $\gamma$-closed, $V_X$ is a clique.
  It follows that~$|V_X| < a$.
  Therefore, there are at most $a \binom{b}{\gamma}$ vertices in $V_i$ adjacent to at least $\gamma$~vertices in $I$.
  For $X \subseteq I$ with $X \ne \emptyset$ and $|X| < \gamma$, let $V_X' = \{ v \in V_i \mid N_G(v) \cap I = X \}$.
  Since~$n > R_\gamma(a, b)$, there exists $X \subseteq I$ of size at most $\gamma - 1$ such that $|V_X'| > R(a, \gamma)$.
  By Ramsey's theorem, we then find an independent set $I' \subseteq V_X'$ of size $\gamma$ (recall that $G$ has no clique of size $a$).
  It follows that $(I \setminus X) \cup I'$ is an independent set of size at least $i$ in $G_i$.
\end{proof}

\begin{corollary}
  Let $\mathcal{G}$ be a class of graphs containing all cliques and independent sets.
  Then, \textsc{Maximum $\mathcal{G}$-Subgraph} has a kernel of size $k^{\Oh(\gamma)}$.
\end{corollary}

\section{Dominating Set in Bipartite and Split Graphs}

We now develop a $k^{\Oh(\gamma)}$-size kernel for \textsc{Dominating Set}
on split graphs and a kernel of size $k^{\Oh(\gamma^2)}$ on graphs with constant clique
size.  \DS{} is defined as follows:

\problemdef
{Dominating Set}
{A graph $G$ and $k \in \mathbb{N}$.}
{Is there a set $D$ of at most $k$ vertices such that $N[D] = V(G)$?}

\DS{} is W[2]-hard when parameterized by $k$, even on split graphs and bipartite graphs \cite{RS08} 
 On the positive side, Alon and Gutner \cite{AG09} proved the fixed-parameter tractability for $k + d$ (recall that $d$ is the degeneracy of the input graph).
 Later, Telle and Villanger \cite{TV12} extended the fixed-parameter tractability to $k + i + j$ on $K_{i, j}$-subgraph-free graphs (that is, graphs that contain no complete bipartite graph $K_{i, j}$ as a subgraph).
 On $c$-closed graphs, Koana et al.~\cite{KKS20} showed that \DS{} is solvable in $\Oh(3^{c/3} + (ck)^{\Oh(k)})$~time.
Furthermore, \DS{} admits kernels of size $\Oh(k^{\min\{ i^2 \!,\, j^2 \}})$ on $K_{i, j}$-subgraph-free graphs when $\min\{ i, j \}$ is constant~\cite{PRS12} and of size~$k^{\Oh(c)}$ on $c$-closed graphs~\cite{KKS20}.
These kernels are essentially optimal:
unless \PHC{}, there is no kernel of size $k^{o(d^2)}$ even in bipartite graphs~\cite{CGH17} or $k^{o(c)}$ even in split graphs~\cite{KKS20}.

\iflong
\vspace*{-2ex}
\fi
\subparagraph{A kernel for bipartite graphs.}

First we obtain a polynomial kernel for \textsc{Dominating Set} on weakly $\gamma$-closed bipartite graphs.
To do so, we exploit the fact that any weakly closed graph is bipartite-subgraph-free whenever its maximum clique size~$\omega$ is fixed:

\begin{lemma}
  \label{lemma:gammaomega}
  Any weakly $\gamma$-closed graph of maximum clique size $\omega$ is $K_{\rho, \rho}$-subgraph-free, where $\rho := \gamma + \omega + 1$.
\end{lemma}
\begin{proof}
  Assume for contradiction that $G$ contains $K_{\rho, \rho}$ as a subgraph.
  Then, there are disjoint vertex sets $U := \{ u_1, \dots, u_{\rho} \}$ and $W := \{ v_1, \dots, v_{\rho} \}$ such that $uw \in E(G)$ for every $u \in U$ and $w \in W$.
  Without loss of generality, assume that $u_i$ (and $w_i$) appears before $u_{i + 1}$ (and $w_{i + 1}$, respectively) for each $i \in [\rho - 1]$ in the closure ordering $\sigma$ of~$G$.
  We can also assume without loss of generality that $u_{\omega + 1}$ appears before $w_{\omega + 1}$ in $\sigma$.
  Then, we have $U' := \{ u_{\omega + 1}, \dots, u_{\rho} \} \subseteq Q(w_{i})$ for every $i \in [\omega + 1]$, and hence we obtain $|Q(w_i) \cap Q(w_{i'})| \ge |U'| = \gamma$ for all $i < i' \in [\omega + 1]$.
  By the definition of the weak closure, $\{ w_1, \dots, w_{\omega + 1} \}$ is a clique.
  This contradicts the assumption that $G$ has no clique of size $\omega + 1$.
\end{proof}

\iflong For a graph $G$, let $\omega_G$ denote the maximum clique size of $G$ (we drop the subscript when clear from context). \fi 
Since \DS{} admits a kernel of size $k^{\min \{ i^2\!, j^2 \}}$ on $K_{i, j}$-subgraph-free graphs, we immediately obtain the following proposition from \Cref{lemma:gammaomega}:

\begin{proposition}
  \DS{} has a kernel of size $k^{\Oh((\gamma + \omega)^2)}$.
\end{proposition}

Since \iflong$\omega_G \le 2$\fi{}  for any bipartite graph $G$, we also obtain the following:

\begin{corollary}
  \DS{} has a kernel of size $k^{\Oh(\gamma^2)}$ in bipartite graphs.
\end{corollary}

\vspace*{-2ex}
\subparagraph{A kernel for split graphs.}

\iflong

\label{sec:dssplit}

We assume that there is no isolated vertices.
For a split graph $G$, there is a bipartition $(C, I)$ of $V(G)$ such that $I$ is a maximum independent set and $C$ is a clique.
We start with a simple observation on \DS{} on split graphs.

\begin{observation}
  \label{obs:dss:cc}
  If $D$ is a solution of $(G, k)$, then there exists a solution $D'$ such that $D' \cap I = \emptyset$.
\end{observation}

We obtain the following reduction rule as a consequence of \Cref{obs:dss:cc}.

\begin{rrule}
  \label{rr:dss:nc}
  If there is a vertex $v \in I$ such that $N(v) = C$, then delete $v$.
\end{rrule}

We will assume that \Cref{rr:dss:nc} is applied exhaustively.
Then, there is a unique bipartition $(C, I)$ of $V(G)$ such that $C$ is a maximum clique and $I$ is a maximum independent set.
We will show that there is a closure ordering $\sigma$ in which every vertex in $C$ appears before any vertex in $I$ (\Cref{lemma:dss:order}).
Before doing so, we prove an auxiliary lemma:

\begin{lemma}
  \label{lemma:dss:orderaux}
  If $\delta(G) \ge \gamma$, then there is a vertex $v \in C$ with $\cl(v) < \gamma$.
\end{lemma}
\begin{proof}
  By the definition of the weak $\gamma$-closure, there exists a vertex $v \in V(G)$ such that $\cl(v) < \gamma$.
  Clearly, the lemma holds for $v \in C$.
  So assume that $v \in I$.
  If there is a vertex $u \in C \setminus N(v)$, then $u$ and $v$ have
  \begin{align*}
    |N(u) \cap N(v)| \ge |N(u) \cap C| = \deg(v) \ge \delta(G) \ge \gamma
  \end{align*}
  common neighbors, contradicting that fact that $\cl(v) < \gamma$.
  Thus, we have $C \setminus N(v) = \emptyset$.
  It follows that $C \cup \{ v \}$ is a clique.
  However, this contradicts the maximality of $C$, and hence~$v \in C$.
\end{proof}

We say that a closure ordering is \emph{good} if every vertex in $C$ appears before any vertex in $I$.
We show the existence of a good closure ordering.

\begin{lemma}
  \label{lemma:dss:order}
  There is a good closure ordering.
\end{lemma}
\begin{proof}
  We prove by induction on the number $n$ of vertices.
  The lemma clearly holds when $\Delta \le \gamma - 1$, as any ordering of $V(G)$ is a closure ordering.
  So assume that $\Delta \ge \gamma$.
  We consider two cases: $\delta \ge \gamma$ and $\delta \le \gamma - 1$.

  If $\delta \ge \gamma$, then there exists a vertex $v \in C$ such that $\cl(v) < \gamma$ by \Cref{lemma:dss:orderaux}.
  By the induction hypothesis, there is a good closure ordering $\sigma'$ of $G - v$.
  Prepending $v$ to $\sigma'$, we obtain a good closure ordering of $G$.

  Now suppose that $\delta \le \gamma - 1$.
  Then, there is a vertex $v \in I$ such that $\deg(v) \le \gamma - 1$.
  Appending $v$ to a closure ordering $\sigma'$ of $G - v$ yields a good closure ordering.
\end{proof}

Note that the proof of \Cref{lemma:dss:order} is constructive.
In fact, one can find a good ordering in polynomial time.
Let us fix a good ordering $\sigma = v_1, \dots, v_n$.
Observe that $N(v) = P(v)$ for every vertex $v \in I$ with respect to $\sigma$.

With a good closure ordering at hand, we aim to develop further reduction rules for \DS{} in split graphs.
To do so, we use the notion of \emph{sunflowers}~\cite{ER60}.
A sunflower with a \emph{core} $T$ is a set family $\mathcal{S} = \{ S_1, \dots, S_k \}$ such that $S_i \cap S_j = T$ for all $i < j \in [k]$.
We say that the sets $S_i \setminus T$ are the \emph{petals} of $\mathcal{S}$.

\begin{lemma}[{\cite[sunflower lemma]{ER60}}]
  \label{lemma:sunflower}
  Let $\mathcal{F}$ be a family of sets, each of size at most $\lambda$.
  If $|\mathcal{F}| \ge \lambda! k^\lambda$, then $\mathcal{F}$ contains a sunflower with at least $k$ petals.
  Moreover, such a sunflower can be found in polynomial time.
\end{lemma}

The sunflower lemma requires the cardinality of sets to be bounded.
Since the vertices of~$G$ could have arbitrary many neighbors, this does not seem to help.
However, by crucially exploiting the weak $\gamma$-closure, we will work around this issue.

To do so, let us introduce further notation.
Recall that $C = \{ v_1, \dots, v_{|C|} \}$ and $I = \{ v_{|C| + 1}, \dots, v_{n} \}$.
For each vertex $u \in I$, let $s_u$ be the smallest $i \in [|C| - 1]$ such that $v_{i + 1} \notin N(u)$ and let $S(u) := N(u) \setminus \{ v_1, \dots, v_{s_u} \}$.
Note that \Cref{rr:dss:nc} ensures that such an integer $s_u$ always exists.

\begin{lemma}
  \label{lemma:dss:sbound}
  For each vertex $u \in I$, it holds that $|S(u)| \le \gamma - 1$.
\end{lemma}
\begin{proof}
  By the definition of $s_u$, there is no edge between $u$ and $v_{s_u + 1}$.
  Thus, we have $|S(u)| = |N(u) \cap Q(v_{s_u + 1})| \le \gamma - 1$.
\end{proof}

Consequently, it follows from \Cref{lemma:sunflower} that the set family $\{ S(v) \mid v \in I \}$ contains a sunflower with at least $k + 2$ petals whenever $|I| \ge (\gamma - 1)! (k + 2)^{\gamma - 1}$.

\begin{rrule}
  \label{rr:dss:sunflower}
  If there is a set $I' \subseteq I$ of $k + 2$ vertices such that $\mathcal{S} := \{ S(v) \mid v \in I' \}$ is a sunflower, then delete $u := \argmax_{u' \in I'} s_{u'}$.
\end{rrule}

\begin{lemma}
  \Cref{rr:dss:sunflower} is correct.
\end{lemma}
\begin{proof}
  Let $G' := G - u$.
  If $D$ is a solution of $(G, k)$ (we can assume that $u \notin D$ by \Cref{obs:dss:cc}), then it is also a solution of $(G', k)$.

  Conversely, suppose that $D'$ is a solution of $(G', k)$.
  Let $C$ be the core of $\mathcal{S}$ and let $A := \{ v_1, \dots, v_{s_u} \} \cup C$.
  We prove by contradiction that $A \cap D' \ne \emptyset$.
  Recall that $N(u') = S(u') \cap \{ v_1, \dots, v_{s_{u'}} \}$ for every vertex $u' \in I' \setminus \{ u \}$.
  Since $s_{u'} < s_u$ (recall that $u = \argmax_{u' \in I'} s_{u'}$), we have $N(u') \setminus A = S(u') \setminus A \subseteq S(u') \setminus C$.
  Thus, the sets $N(u') \setminus A$ are the petals of $\mathcal{S}$ and thus they are pairwise disjoint for $u' \in I' \setminus \{ u \}$.
  If $A \cap D = \emptyset$, then $D$ must contain at least one vertex from $N(u') \setminus A$ for every vertex $u' \in I' \setminus \{ u \}$.
  However, this contradicts the fact that $|D| \le k$, and consequently, we have $A \cap D \ne \emptyset$.
  Since $N(v) \supseteq A$, $D'$ also dominates $u$ in $G$, implying that $D'$ is also a solution of $(G, k)$.
\end{proof}

\Cref{rr:dss:sunflower} gives us an upper bound on the size of $I$.
We use the following reduction rule to obtain an upper bound on the size of $C$.
The correctness follows from \Cref{obs:dss:cc}.

\begin{rrule}
  \label{rr:dss:twin}
  If there are vertices $u, v \in C$ such that $N(u) \supseteq N(v)$, then delete $v$.
\end{rrule}

Finally, we show that the aforementioned reduction rules yield an optimal (up to constants in the exponent) kernel.

\begin{theorem}
  \label{theorem:ds:split}
  \DS{} has a kernel of size $(\gamma k)^{\Oh(\gamma)}$ in split graphs.
\end{theorem}

\begin{proof}
  We apply \Cref{rr:dss:nc,rr:dss:sunflower,rr:dss:twin} exhaustively.
  First, note that $|I| \in (\gamma k)^{\Oh(\gamma)}$, since otherwise \Cref{rr:dss:sunflower} can be further applied by \Cref{lemma:sunflower}.
  Thus, it remains to show that $|C| \in (\gamma k)^{\Oh(\gamma)}$.
  By \Cref{lemma:dss:sbound}, there are at most $(\gamma - 1) |I|$ vertices in~$C$ that are contained in $S(u)$ for some vertex $u \in I$.
  Since $(\gamma - 1) |I| \in (\gamma k)^{\Oh(\gamma)}$, it remains to show that $|C'| \in (\gamma k)^{\Oh(\gamma)}$, where $C' \subseteq C$ is the set of vertices not contained in $S(u)$ for any $u \in I$.
  Let $T := \{ s_u \mid u \in I \} \cup \{ 0, |C| + 1 \}$ and let $0 = t_1 < t_2 \dots < t_{|T|} = |C| + 1 \in T$.
  Note that for every~$v_i\in C$ there exists an integer~$j_i\in[|T|]$ such that~$t_{j_i}\le i< t_{j_i+1}$.  
  Now, we show that for two distinct vertices~$v_i$ and~$v_{i'}$ the integers~$j_i$ and~$j_{i'}$ are pairwise different.
  Assume towards a contradiction that this is not the case.  
  Then, there are two vertices~$v_i, v_{i'} \in C$ with~$j_i=j_{i'}$ and hence $t_{j_i} \le i,i' < t_{j_i+1}$.
  It follows that $N(v_i) = N(v_{i'})$, which contradicts the fact that \Cref{rr:dss:twin} has been exhaustively applied.
  Thus, we have $|C'| \le |T| - 1\le |I| + 1 \in (\gamma k)^{\Oh(\gamma)}$.
\end{proof}

\fi





\newpage
\iflong
\appendix

\section{Extension to a Generalization of Connected Vertex Cover}
\label{sec:coc-kernel}

One might think that to obtain the kernel with respect to~$\gamma+k$ for \ConVC{}{} the property of the remaining set~$V-S$ of being an independent set (for~$S$ being a vertex cover) is necessary. 
In the following, we show that this is not the case.
More precisely, we extend this result to problems where the sizes of connected components in~$V-S$ is at most~$\ell$ for some constant~$\ell$.
Formally, this parameter is defined as follows.
For a graph~$G$ the~$\ell-COC\ number$ is the smallest size of a vertex set~$S$ such that the size of each connected component in~$V(G)-S$ is at most~$\ell$.
Furthermore, an \emph{$\ell$-component set} is a vertex set where each connected component has size at most~$\ell$.
Clearly, the~$1-COC\ number$ is the smallest size of a vertex cover of~$G$, an~$1$-component set is an independent set, and the~$\ell-COC\ numbers$ are monotony decreasing.
In the corresponding \textsc{$\ell-COC$ Deletion} problem one wants to delete at most~$k$ vertices such that each remaining connected component has size at most~$\ell$. 
The \textsc{Connected $\ell-COC$ Deletion} is the corresponding problem such that the deletion set is connected.
In the remainder of this section we provide a kernel of size~$k^{\Oh(\gamma)}$ for \textsc{Connected $\ell-COC$ Deletion} if~$\ell$ is a fixed constant.

To obtain our kernels it is not sufficient anymore to exhaustively apply Reduction Rule~\ref{rr:twincrown} even for~$\ell=2$.
For example consider a graph~$G$ with~$V(G)=\{v_1, \ldots v_{2n}\}$ such that for each~$i\in[n]$ the vertices~$v_{2i-1}$ and~$v_{2i}$ are connected.  
Furthermore, for each odd~$i\ge 3$, vertex~$v_i$ is connected with vertex~$v_1$ and for each even~$i$ vertex~$v_i$ is connected with vertex~$v_2$.
Then,~$G$ does not have any twins,~$2-COC number(G)=2$ but the graph size is not bounded in some function depending only on~$\gamma$ and~$k$.
To this end, we lift the notation of being twins to larger sets of vertices.

\begin{definition}
Let~$G=(V,E)$ be a graph and let~$A,B\subseteq V(G)$ such that~$|A|=\ell=|B|$.
The sets~$A$ and~$B$ are \emph{$\ell$-twins} if and only if there exists an ordering~$a_1, \ldots, a_\ell$ of the vertices of~$A$ and an ordering~$b_1, \ldots , b_\ell$ of the vertices of~$B$ such that for each~$i\in[\ell]$ we have~$N(a_i)\setminus A=N(b_i)\setminus B$.
\end{definition}

Note that the twin relation is identical with the~$1$-twin relation.
Furthermore, observe that the~$\ell$-twin relation is also an equivalence relation which can be computed in $n^{2\ell+\Oh(1)}$~time.
In the following, we prove the generalization of Theorem~\ref{thm:bound-sets-p-and-q} for~$\ell$-twins and~$\ell$-component sets.

\begin{theorem}
\label{thm:bound-sets-p-and-q-generalization}
 Let~$G$ be a graph and let~$D$ be an~$\ell$-component set of~$G$.
  Suppose that for every connected component $Z\in D$, there are at most $t$ other connected components $Z' \in D-Z$ such that~$Z$ and~$Z'$ are~$|Z|$-twins.
  Then, it holds that~$|D| \in t \cdot k^{\Oh(\gamma)}$, where~$k = n - |D|$.
\end{theorem} 

\begin{proof}
  Let $S := V(G) \setminus D$.
  In the following, let~$Z$ be a fixed graph with at most~$\ell$ vertices.
  Furthermore, let~$A$ and~$B$ be two connected components in~$D$ which are isomorphic to~$Z$ with the isomorphism~$f$ which maps the vertices~$a_1,\ldots , a_{|Z|}$ of~$A$ to the vertices of~$B$.  
  In the following, we say that~$A$ and~$B$ are \emph{$P_Z^f$-equivalent}, \emph{$Q_Z^f$-equivalent}, and \emph{$N_Z^f$-equivalent} if~$P(a_i)\setminus A = P(f(a_i))\setminus B$,~$Q(a_i)\setminus A = Q(f(a_i))\setminus B$, and~$N(a_i)\setminus A = N(f(a_i))\setminus B$, respectively.
  Note that since~$A$ and~$B$ are connected components in~$D$ all the above defined sets are contained in~$S$.
  Let~$\mathcal{P}_Z^f$,~$\mathcal{Q}_Z^f$, and~$\mathcal{N}_Z^f$ denote the collection of all~$P_Z^f$-equivalence,~$Q_Z^f$-equivalence, and~$N_Z^f$-equivalence classes, respectively.
  Furthermore, by~$\mathcal{N}$ we denote the union of all classes~$\mathcal{N}_Z^f$ for some graph~$Z$ with at most~$\ell$ vertices and some isomorphism~$f$ between two graphs which are isomorphic to~$Z$.
  Since there is at most one $N_Z^f$-equivalence class for every $P_Z^f$-equivalence class and $Q_Z^f$-equivalence class, we have $|\mathcal{N}_Z^f| \le |\mathcal{P}_Z^f| \cdot |\mathcal{Q}_Z^f|$.
  Furthermore, since there exist at most~$2^{\ell^2}$ graphs with at most~$\ell$ vertices and since there are at most~$\ell!$ many isomorphism for each fixed graph with at most~$\ell$ vertices, we conclude that~$\mathcal{N}\le \ell!\cdot 2^{\ell^2}|\mathcal{N}_Z^f|$ for a graph~$Z$ with at most~$\ell$ vertices and an isomorphism~$f$ which maximizes~$|\mathcal{N}_Z|$.
  By the assumption that there are at most~$t$ vertices in each~$N_Z^f$-equivalence class, we also have~$|D| \le t \cdot |\mathcal{N}|$.
  Since~$\ell$ is a constant, it thus suffices to show that~$|\mathcal{P}_Z^f|, |\mathcal{Q}_Z^f| \in k^{\Oh(\gamma)}$ for each graph~$Z$ with at most~$\ell$ vertices and each corresponding isomorphism~$f$.
  Now, the proof works analog to the proof of Theorem~\ref{thm:bound-sets-p-and-q}.

  First, we prove that $|\mathcal{Q}_Z^f| \in k^{\Oh(\gamma)}$, using the result of Frankl and Wilson (\Cref{prop:fw81}).
To this end, we consider two distinct equivalence classes~$X$ and~$Y$ in~$\mathcal{Q}_Z^f$, let~$A\in D$ be a connected component which is associated with the equivalence class~$X$, and let~$B\in D$ be a connected component which is associated with the equivalence class~$Y$. 
Furthermore, let~$a_1, \ldots , a_{|Z|}$ be a fixed ordering of the vertices of~$A$.
Since~$X$ and~$Y$ are distinct, we conclude by the definition of~$Q_Z^f$-equivalence that there exists an index~$i\in[|Z|]$ such that for~$A_i:=Q(a_i)\setminus A$ and~$B_i:=Q(f(a_i))\setminus B$ we have~$A_i\neq B_i$.
By the definition of weak-closure we observe that~$|A_i\cap B_i|<\gamma$.
Hence,~$|A_i\cap B_i|\in L$ for~$L\in[\gamma-1]$.
Consequently, by \Cref{prop:fw81}, we have $|S|^{\Oh(|L|)} = k^{\Oh(\gamma)}$ possibilities for different neighborhoods in~$S$ of the~$i$th vertex in~$A$.
Since~$|A|\le|Z|$ and~$\ell$ is a constant, there exist~$k^{\Oh(\gamma)\cdot\ell}=k^{\Oh(\gamma)}$ equivalence classes in~$\mathcal{Q}_Z^f$.

  Second, we show that $|\mathcal{P}_Z^f| \in k^{\Oh(\gamma)}$.
  Let~$D_0 := \{ v \in D \mid \exists u, w \in (P(v)\cap S) \colon uw \notin E(G) \}$ be the set of vertices in~$D$ with nonadjacent prior-neighbors in~$S$.
  By the definition of weak $\gamma$-closure, there are at most $\gamma- 1$ vertices of $D_0$ for every pair of nonadjacent vertices in $S$.
  Thus, we have $|D_0| < \gamma \binom{|S|}{2} \in \Oh(\gamma k^2)$.
  Hence, also the set of~$\mathcal{P}_Z^f$-equivalence classes in which there exists a vertex with nonadjacent prior-neighbors is bounded by~$\Oh(\gamma k^2)$.
  
  Let~$D_1:=D\setminus D_0$ be the set of vertices of~$D$ where all prior-neighbors form a clique. 
  Next, we bound the number of equivalence classes in~$\mathcal{P}_Z^f$.
  Since a weakly~$\gamma$-closed graph on~$n$ vertices has~$\Oh(3^{\gamma / 3} n^2)$ maximal cliques~\cite{FRSWW20}, there are~$\Oh(3^{\gamma / 3} k^2)$ equivalence classes $P_Z^f$ such that each graph~$A$ in that class contains a vertex~$v$ such that~$P(v)\setminus A$ constitutes a maximal clique in $G[S]$.
  Now, it remains to bound the number of equivalence classes in~$\mathcal{P}_Z^f$ such that for each graph~$A$ in that class and each vertex~$a_i\in A$ for~$i\in[Z]$ the set~$A_i:=P(a_i)\setminus A$ is a proper subset of some maximal clique~$C_i\subseteq S$.
  Formally,~$A_i\subset C_i\subseteq S$.
  Let~$u_i\in C_i$ be the first vertex in the closure ordering~$\sigma$.
  Since~$A_i\subseteq C_i\subseteq N(u_i)=P(u_i)\cup Q(u_i)$, we have~$A_i=(A_i\cap P(u_i))\cup (A_i\cap Q(u_i))$.
  Since~$A_i\cap P(u_i)=C_i\cap P(u_i)$ by the definition of~$u_i$, we can rewrite~$A_i=(C_i\cap P(u_i))\cap B_i$, where~$B_i:=A_i\cap Q(u_i)$.
  Hence, there is at most one different neighborhood in~$S$ for every maximal clique~$C_i$ in~$G[S]$, vertex~$u_i\in S$, and vertex set~$B_i\subseteq S$ and thereby, we have~$\Oh(3^{\gamma/3})k^2\cdot k\cdot b_i$ many of these neighborhoods, where~$b_i$ denotes the number of choices for~$B_i$.   
Since~$A_i=P(a_i)\setminus A\subseteq N(a_i)$ we conclude that~$B\subseteq Q(u_i)\cap N(a_i)$.
By the definition of weak~$\gamma$-closure we observe that~$|B|\le |Q(u_i)\cap N(a_i)|<\gamma$.
Hence,~$b_i\in k^{\Oh(\gamma)}$.
Overall, for each vertex~$a_i\in A$ we have~$k^{\Oh(\gamma)}$ different neighborhoods.
Thus, the number of such equivalence classes~$\mathcal{P}_Z^f$ is~$k^{\Oh(\gamma)\cdot\ell}=k^{\Oh(\gamma)}$ since~$\ell$ is a constant. 
\end{proof}

To obtain our kernels we generalize Reduction Rules~\ref{rr:isolated} and~\ref{rr:twincrown}.

\begin{rrule}
\label{rr:remove-compos-at-most-ell}
If~$G$ contains a connected component~$Z$ with at most~$\ell$ vertices, then delete~$Z$.
\end{rrule}

The correctness of this rule is obvious and clearly this rule can be exhaustively applied in $\Oh(n^\ell)$~time.

\begin{rrule}
  \label{rr:neighbors-conn-compo}
Let~$T_1, \ldots , T_x\subseteq V(G)$ be a set of~$x$ many~$r$-twins, for some~$r\in[\ell]$.
If~$x\ge k+\ell+2$ , then remove all vertices in~$T_x$ from~$G$.
\end{rrule}

\begin{lemma}
Reduction Rule~\ref{rr:neighbors-conn-compo} is correct.
\end{lemma}
\begin{proof}
Let~$G':=G-T_x$ be the reduced graph. 
Before we prove the correctness, we make the following observation for any connected~$\ell-COC$ set of size at most~$k$.
Since~$|R|\le k$ and~$x\ge k+\ell+2$, we observe that there are at least~$\ell+2$ many~$r$-twins, denoted by~$T_1, \ldots , T_{\ell+2}$, which contain no vertices of~$R$.
Clearly,~$N(T_i)\subseteq R$ for each~$i\in[\ell+2]$ since otherwise~$G-R$ would contain a connected component with at least~$\ell+3$ vertices.

Hence, if~$G$ contains a connected~$\ell-COC$ set~$S$ of size at most~$k$, then also the set~$S':=S\setminus T_x$ is also a connected~$\ell-COC$ set for~$G$ and thus also for~$G'$.
Conversely, if~$G'$ contains a connected~$\ell-COC$ set~$S'$ of size at most~$k$, then by the above argumentation,~$S'$ is also a connected~$\ell-COC$ set of size at most~$k$ for~$G$.
\end{proof}

Note that Reduction Rule~\ref{rr:neighbors-conn-compo} can be exhaustively performed in polynomial time since for each~$r\in[\ell]$ the~$r$-twin relation can be computed in $n^{2r+\Oh(1)}$~time,~$\ell$ is a constant and Reduction Rule~\ref{rr:neighbors-conn-compo} can be applied at most~$n$ times.

Next, we can generalize Theorem~\ref{thm:convc}. 
Theorem~\ref{thm:kernel-l-coc-connected} is a direct consequence of exhaustively applying Reduction Rules~\ref{rr:remove-compos-at-most-ell} and~\ref{rr:neighbors-conn-compo} and the subsequent application of Theorem~\ref{thm:bound-sets-p-and-q-generalization}.

\begin{theorem}
\label{thm:kernel-l-coc-connected}
\textsc{Connected $\ell-COC$ Deletion} in weakly~$\gamma$-closed graphs has a kernel of size~$k^{\mathcal{O}(\gamma)}$.
\end{theorem}
\fi

\end{document}